\def\destin{IJCAI16} 
\newcommand{\content}{Process} 

\documentclass{article}
\usepackage{makeidx}  
\usepackage{authblk}

\usepackage{latexsym,amsmath,amssymb}
\usepackage{amsthm}
\usepackage{comment}
\usepackage{ifthen}
\usepackage{mdwlist,paralist}

\usepackage{epsfig}
\usepackage{epsf}

\usepackage{longtable}
\usepackage{tablefootnote}

\usepackage{url} 

\newboolean{Longversion} \setboolean{Longversion}{false}

\newcommand*{\eqns}[1]{Eq.~#1}
\newcommand*{\eqnsref}[1]{\eqns{\eqref{#1}}} 

\newcommand*{\chapt}[1]{Chapter~#1}
\newcommand*{\sectn}[1]{Section~#1}
\newcommand*{\sectnref}[1]{\sectn{\ref{#1}}} 
\newcommand*{\fig}[1]{Figure~#1}
\newcommand*{\figref}[1]{\fig{\ref{#1}}} 
\newcommand*{\tabl}[1]{Table~#1}
\newcommand*{\tablref}[1]{\tabl{\ref{#1}}} 

\newcommand*{\name}[1]{\protect\emph{#1}} 

\newcommand{\anote}[1]{}
                 
\newcommand{\keywords}[1]{\paragraph{Keywords}\addvspace\baselineskip
\noindent\enspace\ignorespaces#1}

\newcommand{\set}[1]{\left\{ #1 \right\}} 
\newcommand{\paren}[1]{\left( #1 \right)} 
\newcommand{\brackt}[1]{\left[ #1 \right]} 

\newcommand{\inv}[1]{\frac{1}{#1}} 
\newcommand{\abs}[1]{\left| #1 \right|} 
 
\newcommand{\floor}[1]{\left\lfloor {#1} \right\rfloor}

\newcommand{\naturals}{\mathbb{N}} 

\newcommand{\reals}{\mathbb{R}}
\newcommand*{\realsP}{\ensuremath{\mathbb{R}_{+}}}

\newcommand*{\defas}{\ensuremath{\stackrel{\rm \Delta}{=}}}

\DeclareMathOperator*{\Neighb}{N}
\DeclareMathOperator*{\outNeighb}{\Neighb^{+}}

\newcounter{problems}

	{
	    \end{list}
	} 

\newtheorem{theorem}{Theorem}

\newtheorem{defn}{Definition}

\newtheorem{lemma}{Lemma}
\newtheorem{proposition}{Proposition}

\newtheorem{corollary}[theorem]{Corollary}

\newtheorem{example}{Example}

\newcommand{\defined}[1]{\emph{#1}}

\newcommand*{\rmg}{\textrm{RAG}} 

\DeclareMathOperator*{\imp}{act} 
\DeclareMathOperator*{\opin}{opin} 

\DeclareMathOperator*{\got}{got} 

\title{Towards Decision Support in Reciprocation\footnote{An extended abstract is published at~\cite{PolevoydeWeerdtJonker2016}
as ``The Convergence of Reciprocation''.}}

\author[1]{Gleb Polevoy\thanks{g.polevoy@tudelft.nl}}
\author[1]{Mathijs de Weerdt\thanks{M.M.deWeerdt@tudelft.nl}}
\author[1]{Catholijn Jonker\thanks{c.m.jonker@tudelft.nl}}
\affil[1]{Delft University of Technology, Delft, The Netherlands}

\begin{document}


%

\maketitle

\begin{abstract}
People often interact repeatedly:
with relatives, through file sharing, in politics, etc.
Many such interactions are reciprocal: reacting to the actions of
the other.
In order to
facilitate decisions regarding reciprocal interactions,
we analyze the development of reciprocation over time.
To this end, we propose a model for such interactions that is simple enough
to enable formal analysis, but is sufficient to
predict how such interactions will evolve.
Inspired by existing models of international interactions and arguments between spouses,
we suggest a model with two reciprocating attitudes where an agent's
action is a weighted combination of the others' last actions (reacting) and either
i) her innate kindness, or ii) her own last action (inertia).
We analyze a network of repeatedly interacting agents, each having
one of these attitudes, and prove that their actions converge
to specific limits. Convergence means that the interaction stabilizes,
and the limits indicate the behavior after the stabilization.
For two agents, we describe the interaction process
and find the limit values.
For a general connected network, we find these limit values if all
the agents employ the second attitude, and show that the agents' actions then
all become equal. In the other cases, we study
the limit values using simulations.
We discuss how these results predict the development of the interaction and
can be used to help agents decide on their behavior.
\end{abstract}







\section{Introduction}\label{Sec:introd}

Interaction is central in human behavior, e.g.,
at school, in file sharing, in business cooperation and political struggle.
We aim at facilitating decision support for the interacting parties and
for the outside observers. To this end, we want to predict interaction.

Instead of being economically rational, people tend to adopt other
ways of behavior~\cite{RaghunandanSubramanian2012,Rubinstein97}, not necessarily maximizing some
utility function.
Furthermore, people tend to reciprocate, i.e., react on the past actions
of others~\cite{FalkFischbacher2006,FehrGachter2000,GuthSchmittbergerSchwarze1982,Sobel2005}.
Since reciprocation is ubiquitous, predicting it will allow predicting
many real-life interactions and advising on how to improve them.
Therefore, we need a model for reciprocating agents
that is simple enough for analytical analysis and precise enough to predict
such interactions. Understanding such a model would also help understanding how to
improve personal and public good. This is also important for engineering
computer systems that fit human intuition of reciprocity.

Extant models of (sometimes repeated) reciprocation can be classified
as explaining existence or analyzing consequences.
The following models consider the reasons for existence of reciprocal
tendencies, often incorporating evolutionary arguments.
The classical works of Axelrod~\cite{Axelrod1981,Axelrod1984} considered discrete reciprocity
and showed that it is rational for egoists, so that species evolve
to reciprocate.
Evolutionary explanation appears
also in other places, such as~\cite{GuthYaariWitt1992,SethiSomanathan2001},
or~\cite[Chapter~$6$]{Bicchieri2006}, the latter also explicitly considering the
psychological aspects of norm emergence.
In~\cite{VanSegbroeckPachecoLenaertsSantos2012}, they consider pursuing
fairness as a motivation for reciprocation.
In~\cite{AxelrodHamilton1981} and~\cite{fletcherZwick2006},
they considered engendering reciprocation by both the genetical
kinship theory (helping relatives) and
by the utility from cooperating when the same pair of agents interact
multiple times.
The famous work of Trivers~\cite{Trivers1971} showed that sometimes reciprocity is rational,
in much biological detail, and thus, people can evolve to reciprocate.
Gintis~\cite[\chapt{$11$}]{gintis2000} considered discrete actions,
discussing not only the rationally evolved tit-for-tat, but also reciprocity
with no future interaction in sight, what he calls
\defined{strong reciprocity}. He modeled the development of
strong reciprocity.
Several possible reasons for strong reciprocity, such as a social part
in the utility of the agents, expressing itself in emotions,
were considered in \cite{FehrFischbacherGachter2002}.
Berg et al.~\cite{BergDickhautMcCabe1995} proved that people tend to reciprocate
and considered possible motivations,
such as evolutionary stability.
Reciprocal behavior was axiomatically motivated in~\cite{SegalSobel2007},
assuming agents care not only for the outcomes, but also for strategies,
thereby pushed to reciprocate.

On another research avenue,
Given that reciprocal tendencies exist, the following works analyzed what
ways it makes interactions develop.
Some models analyzed reciprocal interactions by defining and analyzing
a game where the
utility function of rational agents directly depends on showing
reciprocation, such as
in~\cite{CoxFriedmanGjerstad2007,DufwenbergKirchsteiger2004,FalkFischbacher2006,Rabin1993}.
The importance of reward/punishment or of incomplete contracts
for the flourishing of reciprocal individuals in the society was shown in~\cite{FehrGachter2000}.

To summarize, reciprocity is seen as an inborn
quality~\cite{FehrFischbacherGachter2002,Trivers1971}, which
has probably been evolved from rationality of agents, as was shown
by Axelrod~\cite{Axelrod1981}.
As we have already said, understanding how a
reciprocal interaction between agents with various reciprocal inclinations
uncurls with time will help explain and predict the dynamics of
reciprocal interaction, such as arms races and personal relations.
This would also be
in the spirit of the call to consider various repercussions of reciprocity
from~\cite{NowakSigmund2005}. Since no analysis considers non-discrete lengthy interaction,
caused by
inborn reciprocation,
(unlike, say, the discrete one from Axelrod~\cite{Axelrod1981,Axelrod1984}),
we model and study how reciprocity makes interaction evolve with time. 

We represent actions by \emph{weight},
where a bigger value means a more desirable contribution or,
in the interpersonal context, investment in the relationship.%
\ifthenelse{\equal{\destin}{IJCAI16}}{
}{
Agents reciprocate both to the agent they are acting on
and to their whole neighborhood.
}%
We model reciprocity
by two reciprocation attitudes, an action's weight being a convex combination between i) one's own kindness or ii) one's own last action,
and the other's and neighborhood's last actions.
The whole past should be considered, but we assume that the
last actions represent the history enough, to facilitate analysis. %
Defining an action (or how much it changes) or a state by a linear combination of the other side's actions and own actions and qualities
was also used to analyze arms race%
~\cite{Dixon1986,Ward1984} and spouses' interaction~\cite{GottmanSwansonMurray1999} (piecewise linear in this case).
Attitude i) depending on the (fixed) kindness is called \name{fixed}, and ii) depending on one's own last action is called
\name{floating}. Given this model, we study its behavioral repercussions.

There are several reminiscent but different models.
The \name{floating} model resembles opinions that converge to
a consensus~\cite{BlondelHendrickxOlshevskyTsitsiklis2005,Moreau2005,TsitsiklisBertsekasAthans1986,DeGroot1974},
while the \name{fixed} model resembles converging to a general equilibrium
of opinions~\cite{BindelKleinbergOren2011}. Of course, unlike the models
of spreading opinions, we consider different actions on various neighbors,
determined by direct reaction and a reaction to the whole neighborhood.
Still, because of some technical reminiscence to some of our
models, we do use those for one of our proofs.
Another similar model is that of
monotonic concession~\cite{RosenscheinZlotkin1994} and that of
bargaining over dividing a pie~\cite{Rubinstein1982}. The main difference
is that in those models, the agents decide what to do, while in our case,
they follow the reciprocation formula.

%
\begin{example}\label{ex:colleages}
Consider $n$ colleagues~$1, 2, \ldots, n$,
who can help
or
harm each other.
Let the possible actions be: giving bad
work, showing much contempt, showing little contempt,
supporting emotionally a little, supporting emotionally a lot,
advising, and let their respective weight be a point in%
~$[-1, -0.5)$, $[-0.5, -0.2)$, $[-0.2. 0)$, $(0, 0.4)$, $[0.4, 0.7)$, $[0.7, 1]$.
Assume that each person knows what the other did to him last time.
The social climate, meaning what the whole group did, also influences behavior.
However, we may just concentrate on a single pair of even-tempered colleagues
who reciprocate regardless the others.
\end{example}

To understand and predict reciprocal behavior, we look at the limit of time approaching infinity,
since this describes what actions will take place from some time on.
We first consider two agents in \sectnref{Sec:dynam_interact_pair},
assuming their interaction is independent of other agents,
or that the total influence of the others
on the pair is negligible.
This assumption allows for deeper a theoretical
analysis of the interaction than in the general case. 
The values in the limit for two agents will be also implied by a
general convergence result that is presented later, unless both agents are \name{fixed}. We
still present them with the other results for two agents for the completeness of
\sectnref{Sec:dynam_interact_pair}. 
\sectnref{Sec:dynam_interact_interdep} studies interaction of many agents,
where the techniques we used for two agents are not applicable,
and we show exponentially fast convergence.
Exponential convergence means a rapid stabilizing, and it explains acquiring personal behavioral styles,
which is often seen in practice~\cite{RobertsWaltonViechtbauer2006}.
We find the limit
when all the agents act synchronously and at most one has the \name{fixed} reciprocation attitude.
Among other things, we prove that when at most one agent is \name{fixed},
the limits of the actions of all agents 
are the same, explaining formation of organizational
subcultures, known in the literature~\cite{Hofstede1980}.
We also find that only the kindness values of the \name{fixed} agents influence
the limits of the various actions,
thereby explaining that persistence (i.e., being faithful to one's inner inclination) makes interaction go
one's own way, which is reflected in daily life in the
recommendations to reject undesired requests by
firmly repeating the reasons for rejection%
~\cite[\chapt{1}]{BreitmanHatch2000} and~\cite[Chapter~$8$]{Ury2007}.
Other cases are simulated in \sectnref{Sec:dynam_interact_interdep_sim}.
These results describe the interaction process and lay the foundation for further analysis of
interaction.

The major contributions are proving convergence and finding its limits
for at most one \name{fixed} agent or for two agents.
These allow to explain the above mentioned phenomena and predict
reciprocation. The predictions can assist in deciding whether
a given interaction will be profitable, and in engineering
more efficient multi-agent systems,
fitting the reciprocal intuition of the users.

\section{Modeling Reciprocation}\label{Sec:formal_model}

\ifthenelse{\NOT \equal{\destin}{AAMAS16}}{
\subsection{Basic}\label{Sec:formal_model:basic_facts}
}{
}%

Let $N = \set{1, 2, \ldots, n}$ be $n \geq 2$ interacting agents. 
We assume that possible actions are described by an undirected interaction graph
$G = (N, E)$, such that agent $i$ acts on $j$ and vice versa
if and only if $(i, j) \in E$.
Denote the degree of agent $i \in N$ in $G$ by $d(i)$.
This allows for various topologies, including heterogeneous ones,
like those in~\cite{SantosPachecoLenaerts2006}.
To be able to mention directed edges, we shall treat this graph
as a directed one, where for every $(i, j) \in E$, we have $(j, i) \in E$.
Time is modeled by a set of discrete moments
$t \in T \defas \set{0, 1, 2, \ldots}$, defining a time slot whenever
at least one agent acts.
Agent $i$ acts at times $T_i \defas \set{t_{i, 0} = 0, t_{i, 1}, t_{i, 2}, \ldots} \subseteq T$, and 
$\cup_{i \in N} T_i = T$.
We assume that all agents act at $t = 0$, since otherwise we cannot
sometimes consider the last action of another agent, which would force us to
complicate the model and render it even harder for theoretical analysis.
%
When all agents always act at the same times~($T_1 = T_2 = \ldots = T_n = T$), we
say they act \defined{synchronously}.

For the sake of asymptotic analysis,
we assume that each agent gets to act an infinite number of times;
that is, $T_i$ is infinite for every $i \in N$. Any real application will,
of course, realize only a finite part of it, and infinity models
the unboundedness of the process in time.

When $(i, j)$ is in $E$,
we denote the weight of an action by agent $i \in N$ on
another agent $j \in N$ at moment $t$ by $\imp_{i, j}(t) \colon T_i \to \reals$.
We extend $\imp_{i, j}$ to $T$ by assuming that
at $t \in T \setminus T_i$, we have $\imp_{i, j}(t) = 0$.
Since only the weight of an action is relevant, we usually
write ``action'' while referring to its weight.
For example, when interacting by file sharing, sending a valid
piece of a file, nothing, or a piece with a virus has a positive, zero,
or a negative weight, respectively.


For $t \in T$, we define \defined{the last action time
$s_i(t) \colon T \to T_i$ of agent $i$} as the largest $t' \in T_i$ that is at most $t$. 
Since $0 \in T_i$, this is well defined.
The last action of agent $i$ on (another) agent $j$ is given by $x_{i, j}(t) \defas \imp_{i, j}(s_i(t))$.
Thus, we have defined $x_{i, j}(t) \colon T \to \reals$,
and we use mainly this concept $x_{i, j}$ in the paper.
We denote the total received contribution from all the neighbors $\Neighb(i)$ at their last action times not later than $t$ by
$\got_i(t) \colon T \to \reals$; formally,
$\got_i(t) \defas \sum_{j\in \Neighb(i)}{x_{j, i}(t)}$.


We now define two reciprocation attitudes, which define how an agent
reciprocates. We need the following notions.
The kindness of agent $i$ is denoted by $k_i \in \reals$;
w.l.o.g., $k_n \geq \ldots \geq k_2 \geq k_1$ throughout the paper.
Kindness models inherent inclination to help others; in particular,
it determines the first action of an agent, before others have acted.
We model agent $i$'s inclination to mimic a neighboring agent's action and
the actions of the whole neighborhood in $G$ by reciprocation coefficients
$r_i \in \brackt{0, 1}$
and $r'_i \in \brackt{0, 1}$ respectively,
such that $r_i + r'_i \leq 1$. Here, $r_i$ is the fraction of $x_{i, j}(t)$
that is determined by the last action of $j$ upon $i$, and $r_i'$ is the
fraction that is determined by $\frac{1}{\abs{\Neighb(i)}}$th of
the total contribution to $i$ from all the neighbors at the last time.
\ifthenelse{\NOT \equal{\destin}{AAMAS16}}{
\subsection{Reciprocation}\label{Sec:formal_model:recip}
}{
}%
Intuitively,
the \name{fixed} attitude depends on the agent's kindness at every action,
while the \name{floating} one is loose, moving freely in the
reciprocation process, and kindness directly influences such behavior only at $t = 0$. In both cases $x_{i,j}(0) \defas k_i$.
\begin{defn}\label{def:fix_recip}
For the \defined{fixed reciprocation attitude}, agent $i$'s reaction
on the other agent $j$ and on the neighborhood is determined by the agent's kindness weighted
by~$1 - r_i - r'_i$, by the other agent's action weighted by~$r_i$
and by the total action of the neighbors weighted by~$r'_i$ and divided over all the neighbors:
That is, for $t \in T_i$,
$\imp_{i,j}(t) = x_{i,j}(t) \defas$
\begin{eqnarray*}
(1 - r_i - r'_i) \cdot k_i + r_i \cdot x_{j, i}(t-1)
+ r'_i \cdot \frac{\got_i(t - 1)}{\abs{\Neighb(i)}}.
\end{eqnarray*}
\end{defn}

\begin{defn}\label{def:float_recip}
In the \defined{floating reciprocation attitude}, agent $i$'s action
is a weighted average of her own last action, of that of
the other agent $j$ and of the total action of the neighbors divided over all the neighbors:
To be precise, for $t \in T_i$,
$\imp_{i,j}(t) = x_{i, j}(t) \defas$
\begin{eqnarray*}
(1 - r_i - r'_i) \cdot x_{i, j}(t-1) + r_i \cdot x_{j, i}(t-1) 
+ r'_i \cdot \frac{\got_i(t - 1)}{\abs{\Neighb(i)}}.
\end{eqnarray*}
\end{defn}

%
The relations are (usually inhomogeneous) linear recurrences with
constant coefficients. We could express the dependence
$x_{i, j}(t)$ only on $x_{i, j}(t')$ with $t' < t$, but then the
coefficients would not be constant, besides the case of two \name{fixed} agents.
We are not aware of a method to use the
general recurrence theory to improve our results.

\ifthenelse{\NOT \equal{\destin}{AAMAS16}}{
\subsection{Clarifications}
}{
}%

Compared to the other models, our model takes reciprocal actions
as given and looks at the process, while other models either
consider how reciprocation originates, such as the evolutionary
model of Axelrod~\cite{Axelrod1981},
or take it as given and consider specific games, such as in
\cite{CoxFriedmanGjerstad2007,DufwenbergKirchsteiger2004,FalkFischbacher2006,Rabin1993}.

In~Example~\ref{ex:colleages}, let (just here) $n = 3$ and the reciprocation coefficients
be~$r_1 = r_2 = 0.5, r_1' = r_2' = 0.3, r_3 = 0.8, r_3' = 0.1$.
Assume the kindness to be~$k_1 = 0, k_2 = 0.5$ and $k_3 = 1$.
Since this is a small group, all the colleagues may interact, so the graph
is a clique\footnote{A clique is a fully connected graph.}. At $t = 0$, every agent's action on every other agent
is equal to her kindness value, so agent $1$ does nothing, agent $2$
supports emotionally a lot, and $3$ provides advice.
If all agents act synchronously, meaning $T_1 = T_2 = T_3 = \set{0, 1, \ldots}$ , and all
get carried away by the process, meaning that they forget the kindness
in the sense of employing \name{floating} reciprocation, then, at $t = 1$ they act
as follows:
$x_{1, 2}(1)
= (1 - 0.5 - 0.3) \cdot 0 + 0.5 \cdot 0.5 + 0.3 \cdot \frac{0.5 + 1}{2} = 0.475$
(supports emotionally a lot),
$x_{1, 3}(1)
= (1 - 0.5 - 0.3) \cdot 0 + 0.5 \cdot 1 + 0.3 \cdot \frac{0.5 + 1}{2} = 0.975$
(provides advice),
$x_{2, 1}(1)
= (1 - 0.5 - 0.3) \cdot 0.5 + 0.5 \cdot 0 + 0.3 \cdot \frac{0 + 1}{2} = 0.25$
(supports emotionally a little),
and so on.

Consider modeling tit for tat~\cite{Axelrod1984}:
\begin{example}
In our model, the tit for tat with two options, - cooperate or
defect, is easily modeled with $r_i = 1$, $k_i = 1$, meaning
that the original action is cooperating ($1$) and the next action is
the current action of the other player.
Since we consider a mechanism,
rather than a game, the agents will always cooperate. If one agent begins
with cooperation ($k_1 = 1$) and the other one with defection ($k_2 = 0$),
acting synchronously,
then they will alternate.
\end{example}

\ifthenelse{\NOT\equal{\destin}{AAMAS16}}{
The notation is summarized in \tablref{tbl:notation}.

\begin{table}[ht]
\begin{tabular}{|p{0.17\textwidth}|p{0.8\textwidth}|}
\hline
Term: & Meaning:\\
\hline
$\imp_{i, j}(t) \colon T \to \reals$ & The action of $i$ on another agent $j$ at time $t$.\\
\hline
$T_i$ & The time moments when agent $i$ acts.\\
\hline
Synchronous & $T_1 = T_2 = \ldots = T_n$.\\
\hline
$s_i(t) \colon T \to T_i$ & $\max\set{t' \in T_i | t' \leq t}$.\\
\hline
$x_{i, j}(t) \colon T \to \reals$ & $\imp_{i, j}(s_i(t))$.\\
\hline
$\got_i(t) \colon T \to \reals$ & $\sum_{j\in \Neighb(i)}{x_{j, i}(t)}$.\\
\hline
$k_i$ & The kindness of agent $i$.\\
\hline
$r_i, r'_i \in [0, 1], r_i + r'_i \leq 1$ & The reciprocation coefficients of agent $i$.\\
\hline
Agent $i$ has the \name{fixed} reciprocation attitude, $j$ is another agent
& At moment $t \in T_i$,
\begin{eqnarray*}
x_{i,j}(t) \defas
\begin{cases}
(1 - r_i - r'_i) \cdot k_i + r_i \cdot x_{j, i}(t-1)\\
+ r'_i \cdot \frac{\got_i(t - 1)}{\abs{\Neighb(i)}} & t > t_{i, 0}\\
k_i & t = t_{i, 0} = 0.
\end{cases}
\end{eqnarray*}\\
\hline
Agent $i$ has the \name{floating} reciprocation attitude, $j$ is another agent
& At moment $t \in T_i$,
\begin{eqnarray*}
x_{i, j}(t) \defas
\begin{cases}
(1 - r_i - r'_i) \cdot x_{i, j}(t-1)\\ + r_i \cdot x_{j, i}(t-1) 
+ r'_i \cdot \frac{\got_i(t - 1)}{\abs{\Neighb(i)}} 
 & t > t_{i, 0}\\
k_i & t = t_{i, 0} = 0.
\end{cases}
\end{eqnarray*}\\
\hline
\end{tabular}
\caption{The notation used throughout the paper.}
\label{tbl:notation}
\end{table}
}{
}%

\section{Pairwise Interaction}\label{Sec:dynam_interact_pair}%
\ifthenelse{\equal{\content}{Process} \OR \equal{\content}{All}}{

We now consider an interaction of two agents, $1$ and $2$,
since this assumption allows proving much more than we will be able
to in the general case.
The model reduces to a pairwise interaction,
when
$r'_i = 0$ or when there are no neighbors besides
the other agent in the considered pair.
We assume both, w.l.o.g.
\ifthenelse{\NOT \equal{\destin}{IJCAI16}}{
When $T_1$ contains precisely
all the even numbered slots and $T_2$ zero and all the odd ones, we
say they are \defined{alternating}.
}{
}%
Since agent $1$ can only act on agent $2$ and vice versa, we write
$\imp_i(t)$ for $\imp_{i, j}(t)$,
$x(t)$ for $x_{1, 2}(t)$ and $y(t)$ for $x_{2, 1}(t)$.
%
}{
}%

We analyze
the case of both agents being
\name{fixed}, then the case of the \name{floating},
and then the case where one is \name{fixed}
and the other one is \name{floating}.
To formally discuss the actions after the interaction
has stabilized, we consider the limits (if exist)\footnote{Agent $i$ acts at the times in $T_i = \set{t_{i, 0} = 0, t_{i, 1}, t_{i, 2}, \ldots}$.}
$\lim_{p \to \infty}{\imp_1(t_{1, p})}$,
and $\lim_{t \to \infty}{x(t)}$, for agent $1$,
and $\lim_{p \to \infty}{\imp_2(t_{2, p})}$
and $\lim_{t \to \infty}{y(t)}$ for agent $2$.
Since the sequence $\set{x(t)}$ is $\set{\imp_1(t_{1, p})}$ with finite repetitions,
the limit $\lim_{t \to \infty}{x(t)}$ exists if and only if $\lim_{p \to \infty}{\imp_1(t_{1, p})}$ does. If they exist, they are equal;
the same holds for $\lim_{t \to \infty}{y(t)}$ and $\lim_{p \to \infty}{\imp_2(t_{2, p})}$.
Denote $L_x \defas \lim_{t \to \infty}{x(t)}$ and $L_y \defas \lim_{t \to \infty}{y(t)}$.

\subsection{\name{Fixed} Reciprocation}

Here we prove that both action sequences converge.
\begin{theorem}\label{The:fixed_recip}
If the reciprocation coefficients are not both $1$, which means $r_1 r_2 < 1$, then we have, for $i \in N$:
$\lim_{p \to \infty} {\imp_i(t_{i, p})} = \frac{(1 - r_i) k_i + r_i (1 - r_j) k_j}{1 - r_i r_j}$.
\end{theorem}
The assumption that not both reciprocation coefficients are $1$ and
the similar assumptions in the following theorems (such as $1 > r_i > 0$)
mean that the agent neither ignores the other's action, nor does it
copy the other's action. These are to be expected in real life.
\ifthenelse{\NOT \equal{\destin}{IJCAI16}}{
The limits of these actions are shown
in Figures~\ref{fig:fixed_fixed_limit_agent1}
and \ref{fig:fixed_fixed_limit_agent2}.
\begin{figure}
\centering%

\begin{minipage}{.47\textwidth}
\includegraphics[width=0.98\textwidth]{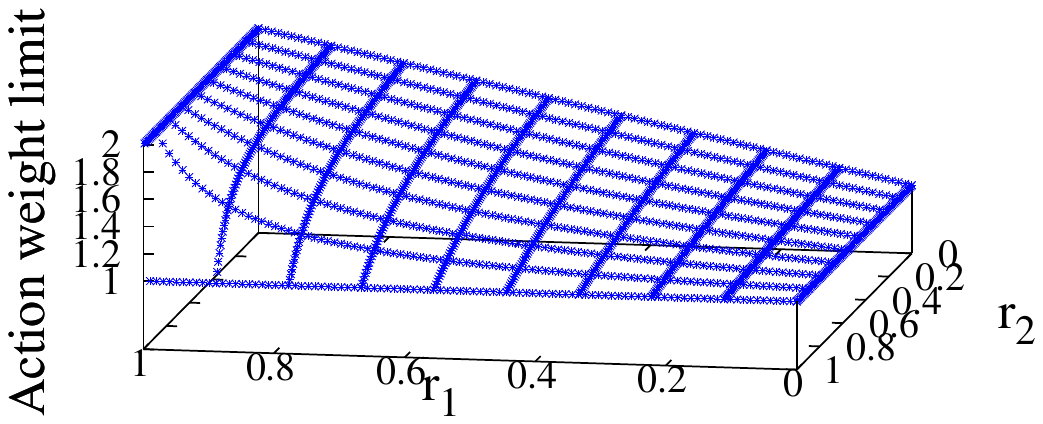}

\protect\caption{The limit of the actions of agent $1$ as a function of the reciprocity coefficients,
for a \name{Fixed} - \name{fixed} reciprocation, $k_1 = 1, k_2 = 2$.
Given $r_1$, agent $2$ receives most when $r_2 = 0$.}
\label{fig:fixed_fixed_limit_agent1}%

\end{minipage}%
\hspace{0.04\textwidth}%
\begin{minipage}{.47\textwidth}
\centering
\includegraphics[width=0.98\textwidth]{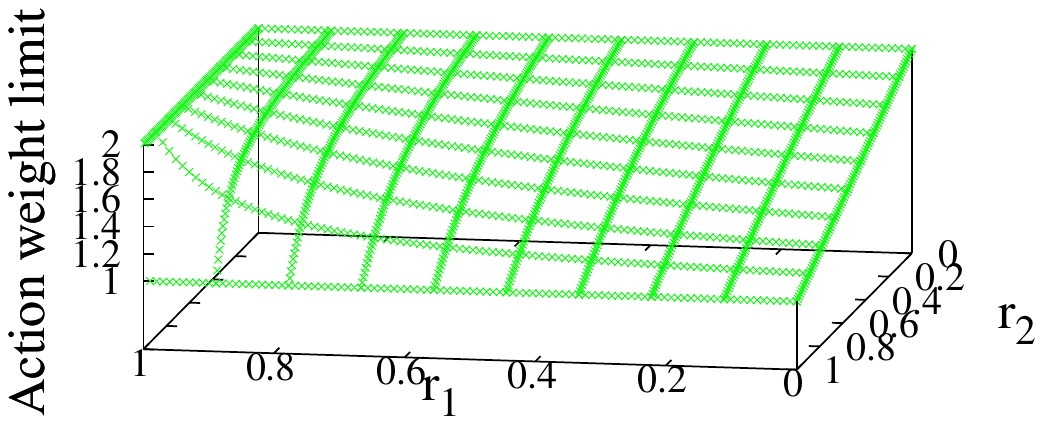}

\protect\caption{The limit of the actions of agent $2$ as a function of the reciprocity coefficients,
for a \name{Fixed} - \name{fixed} reciprocation, $k_1 = 1, k_2 = 2$.
Given $r_2$, agent $1$ receives most when $r_1 = 1$.}
\label{fig:fixed_fixed_limit_agent2}%

\end{minipage}

\end{figure}
}{
}%
In Example~\ref{ex:colleages}, if agents $1$ and $2$ employ \name{fixed}
reciprocation, $r_1 = r_2 = 0.5, r_1' = r_2' = 0.0$ and
$k_1 = 0, k_2 = 0.5$, then we obtain
$L_x = \frac{0.5 \cdot (1 - 0.5) 0.5}{1 - 0.5 \cdot 0.5}
= 1/6$ and
$L_y = \frac{(1 - 0.5) \cdot 0.5}{1 - 0.5 \cdot 0.5}
= 1/3$.

In order to prove this theorem, we first show that it is sufficient to analyze the
synchronous case, i.e.,~$T_1 = T_2 = T$.
\begin{lemma}\label{lemma:fixed_recip:subseq_synch_case}
Consider a pair of interacting agents.
Denote the action sequences in case both agents acted at the same
time, (i.e.,~$T_1 = T_2 = T$), by $\set{x'(t)}_{t \in T}$ and $\set{y'(t)}_{t \in T}$, respectively.
Then the action sequences\footnote{Agent $i$ acts at the times in $T_i = \set{t_{i, 0} = 0, t_{i, 1}, t_{i, 2}, \ldots}$.} $\set{\imp_1(t_{1,p})}_{p \in \naturals}, \set{\imp_2(t_{2,p})}_{p \in \naturals}$ are subsequences
of $\set{x'(t)}_{t \in T}$ and $\set{y'(t)}_{t \in T}$, respectively.
\end{lemma}%
\ifthenelse{\NOT \equal{\destin}{IJCAI16}}{
\begin{proof}
We prove by induction that for each $p > 0$, the sequence
$\imp_1(t_{1, 0}), \imp_1(t_{1, 1}), \ldots, \imp_1(t_{1, p})$ is a subsequence
of the
$x'(0), x'(1), \ldots, \imp_1(t_{1, p})$
and the sequence
$\imp_2(t_{2, 0}), \imp_2(t_{2, 1}), \ldots, \imp_2(t_{2, p})$ is a subsequence
of the
$y'(0), y'(1), \ldots, y'(t_{2, p})$.

For $p = 0$, this is immediate, since $t_{i, 0} = 0$.

For the induction step, assume that the lemma holds for $p - 1$ and prove
it for $p > 0$. By definition,
$\imp_1(t_{1, p}) = (1 - r_1) \cdot k_1 + r_1 \cdot \imp_2(s_2(t_{1, p}-1))$,
and since by the induction hypothesis, $\imp_2(s_2(t_{1, p}-1))$ is
an element in the sequence $\set{y'(t)}$, we conclude that $\imp_1(t_{1, p})$ is an
element in the sequence $\set{x'(t)}$. Moreover, in $\set{x'(t)}$ this element comes after
$\imp_1(t_{1, {p - 1}})$, because either
$\imp_1(t_{1, {p - 1}}) = \imp_1(0)$ or
$\imp_1(t_{1, {p - 1}}) = (1 - r_1) \cdot k_1 + r_1 \cdot \imp_2(s_2(t_{1, {p - 1}}-1))$
and $\imp_2(s_2(t_{1, {p - 1}}-1))$ precedes $\imp_2(s_2(t_{1, p}-1))$
in $\set{y'(t)}$ by the induction hypothesis. This proves the induction step for
agent $1$, and it is proven by analogy for $2$.
\end{proof}
}{
The proof follows from Definition~\ref{def:fix_recip} by induction. (Straightforward proofs in this paper have been replaced by their general ideas due to lack of space).
}%
Using this lemma, it is sufficient to further assume
the synchronous case.
\begin{lemma}\label{lemma:fixed_recip:per_seq_behav}
In the synchronous case,
for every $t > 0: x(2t - 1) \geq x(2t + 1)$, and for every
$t \geq 0: x(2t) \leq x(2t + 2) \leq x(2t + 1)$.
By analogy,
$\forall t > 0: y(2t - 1) \leq y(2t + 1)$, and
$\forall t \geq 0: y(2t) \geq y(2t + 2) \geq y(2t + 1)$.
All the inequations are strict if and only if
$0 < r_1, r_2 < 1, k_2 > k_1$.\footnote{We always assume that $k_2 \geq k_1$.}
\end{lemma}
Since we also have $t \geq 0: x(2t) \leq x(2t + 1)$, we
obtain $t > 0: x(2t - 1) \geq x(2t + 1) \geq x(2t)$, and for every
$t \geq 0: x(2t) \leq x(2t + 2) \leq x(2t + 1)$.
By analogy,
$\forall t > 0: y(2t - 1) \leq y(2t + 1) \leq y(2t)$, and
$\forall t \geq 0: y(2t) \geq y(2t + 2) \geq y(2t + 1)$.
Intuitively, this means that the sequence $\set{x(t)}$ is alternating while
its amplitude is getting smaller, and the same holds for the sequence $\set{y(t)}$, with another
alternation direction.
The intuitive reasons are that first, agent $1$ increases
her action, while $2$ decreases it. Then, since $2$ has decreased
her action, so does $1$, while since $1$ has increased hers, so
does $2$.%
\ifthenelse{\NOT \equal{\destin}{IJCAI16}}{
This alternating process is going on, and the amplitude
subsides, since the changing part is convexly combined with the
constant one.

To prove the theorem we use only the monotonicity of the
subsequences of the even and of the odd actions.
}{
}%
We now prove the lemma.
\begin{proof}
We employ induction. For $t = 0$, we need to show that
$x(0) \leq x(2) \leq x(1)$ and $y(0) \geq y(2) \geq y(1)$.
We know that $x(0) = k_1$,
$x(1) = (1 - r_1) \cdot k_1 + r_1 \cdot k_2$,
and $y(0) = k_2$,
$y(1) = (1 - r_2) \cdot k_2 + r_2 \cdot k_1$.
Since $y(1) \leq k_2$, we have
$x(2) = (1 - r_1) \cdot k_1 + r_1 \cdot y(1) \leq x(1)$.
Since $y(1) \geq k_1$, we also have
$x(2) = (1 - r_1) \cdot k_1 + r_1 \cdot y(1) \geq x(0)$.
The proof for $y$s is analogous.

For the induction step, for any $t > 0$, assume that the lemma holds
for $t - 1$, which means $x(2t - 3) \geq x(2t - 1)$ (for $t > 1$),
$x(2t - 2) \leq x(2t) \leq x(2t - 1)$, and
$y(2t - 3) \leq y(2t - 1)$ (for $t > 1$),
$y(2t - 2) \geq y(2t) \geq y(2t - 1)$.

We now prove the lemma for $t$.
By Definition~\ref{def:fix_recip}, $x(2t - 1) = (1 - r_1)k_1 + r_1 y(2t - 2)$ and
$x(2t + 1) = (1 - r_1)k_1 + r_1 y(2t)$.
Since $y(2t - 2) \geq y(2t)$, we have $x(2t - 1) \geq x(2t + 1)$.
By analogy, we can prove that $y(2t - 1) \leq y(2t + 1)$.

Also by definition, $x(2t) = (1 - r_1)k_1 + r_1 y(2t - 1)$ and
$x(2t + 2) = (1 - r_1)k_1 + r_1 y(2t + 1)$.
Since $y(2t - 1) \leq y(2t + 1)$, we have $x(2t) \leq x(2t + 2)$.
By definition, $x(2t + 1) = (1 - r_1)k_1 + r_1 y(2t)$.
Since $y(2t) \geq y(2t - 1)$, we conclude that
$x(2t + 1) \geq x(2t)$. By analogy, we prove that
$y(2t + 1) \leq y(2t)$. From this, we conclude that
$x(2t + 2) \leq x(2t + 1)$, and we have shown that
$x(2t) \leq x(2t + 2) \leq x(2t + 1)$.
By analogy, we prove that $y(2t) \geq y(2t + 2) \geq y(2t + 1)$.

The equivalence of strictness in all the inequations to
$0 < r_1, r_2 < 1, k_2 > k_1$ is proven by repeating the proof
with  strict inequalities in one direction, and by noticing that
not having one of the conditions
$0 < r_1, r_2 < 1, k_2 > k_1$ implies
equality in at least one of the statements of the lemma.
\end{proof}

With these results we now prove Theorem~\ref{The:fixed_recip}.
\begin{proof}
Using Lemma~\ref{lemma:fixed_recip:subseq_synch_case}, we
assume the synchronous case.
We first prove convergence, and then find its limit.
For each agent, Lemma~\ref{lemma:fixed_recip:per_seq_behav} implies
that the even actions form a monotone sequence,
and so do the odd ones. Both sequences are bounded, which can be easily
proven by induction, and therefore each
one converges. The whole sequence converges if and only if both limits
are the same. We now show that they are indeed the same for
the sequence $\set{x(t)}$; the proof for $\set{y(t)}$ is analogous.%
\ifthenelse{\equal{\destin}{IJCAI16}}{
$ x(t + 1) - x(t)\\ = (1 - r_1)k_1 + r_1 y(t) - (1 - r_1)k_1 - r_1 y(t - 1)\\
= r_1 (y(t) - y(t - 1)) =  r_1 r_2 (x(t - 1) - x(t - 2)) = \ldots\\
= (r_1 r_2)^{\floor{t / 2}}
\begin{cases}
x(1) - x(0) & t = 2 s, s \in \naturals\\
x(2) - x(1) & t = 2 s + 1, s \in \naturals.
\end{cases}$
}{
\begin{align*}
 &x(t + 1) - x(t)\\ &\quad= (1 - r_1)k_1 + r_1 y(t) - (1 - r_1)k_1 - r_1 y(t - 1)\\
&\quad= r_1 (y(t) - y(t - 1)) =  r_1 r_2 (x(t - 1) - x(t - 2)) = \ldots\\
&\quad= (r_1 r_2)^{\floor{t / 2}}
\begin{cases}
x(1) - x(0) & t = 2 s, s \in \naturals\\
x(2) - x(1) & t = 2 s + 1, s \in \naturals.
\end{cases}
\end{align*}
}%
As $r_1 r_2 < 1$, this difference goes to 0 as $t$ goes to $\infty$.
Thus, $x(t)$ converges (and so does $y(t)$).
To find the limits $L_x = \lim_{t \to \infty}{x(t)}$
and $L_y = \lim_{t \to \infty}{y(t)}$, notice that in the limit we have
$(1 - r_1) k_1 + r_1 L_y = L_x$ and 
$(1 - r_2) k_2 + r_2 L_x = L_y$
with the unique solution:
$L_x = \frac{(1 - r_1) k_1 + r_1 (1 - r_2) k_2}{1 - r_1 r_2}$ and 
$L_y = \frac{(1 - r_2) k_2 + r_2 (1 - r_1) k_1}{1 - r_1 r_2}$.
\end{proof}

\ifthenelse{\NOT \equal{\destin}{IJCAI16}}{
If, unlike the theorem assumes, $r_1 r_2 = 1$, then
since $r_1 r_2  = 1 \iff r_1 = r_2 = 1$,
each agent just repeats what the other one did last time, thereby
changing roles. It means, in particular, that unless $k_1 = k_2$,
no convergence takes place.
}{
}%

We see that $L_x \leq L_y$,
which is intuitive, since the agents are always considering their kindness,
so the kinder one acts with a bigger weight also in the limit.
\ifthenelse{\NOT \equal{\destin}{IJCAI16}}{
When this limit inequality is strict, we have $x(t) < y(t)$ for all
$t \geq t_0$ for some $t_0$. To find when it is strict, consider
the following:
\begin{eqnarray*}
\frac{(1 - r_1) k_1 + r_1 (1 - r_2) k_2}{1 - r_1 r_2} = \frac{(1 - r_2) k_2 + r_2 (1 - r_1) k_1}{1 - r_1 r_2}\\
\iff (1 - r_1)(1 - r_2) k_1 = (1 - r_1)(1 - r_2) k_2\\
\iff r_1 = 1 \lor r_2 = 1 \lor k_1 = k_2,
\end{eqnarray*}
and thus, it is strict if and only if $r_1 < 1 \land r_2 < 1 \land k_1 < k_2$.
}{
}%
In the simulation of the actions over time in Figure~\ref{fig:fixed_fixed_03_05_09},
on the left, $y(t)$ is always larger than $x(t)$, and on the right,
they alternate several times before $y(t)$ becomes larger.
%
\begin{figure}%

\begin{minipage}{.47\textwidth}
\centerline{%
\includegraphics[trim = 35mm 80mm 40mm 90mm, clip, width=0.5\textwidth, height=1.5in]{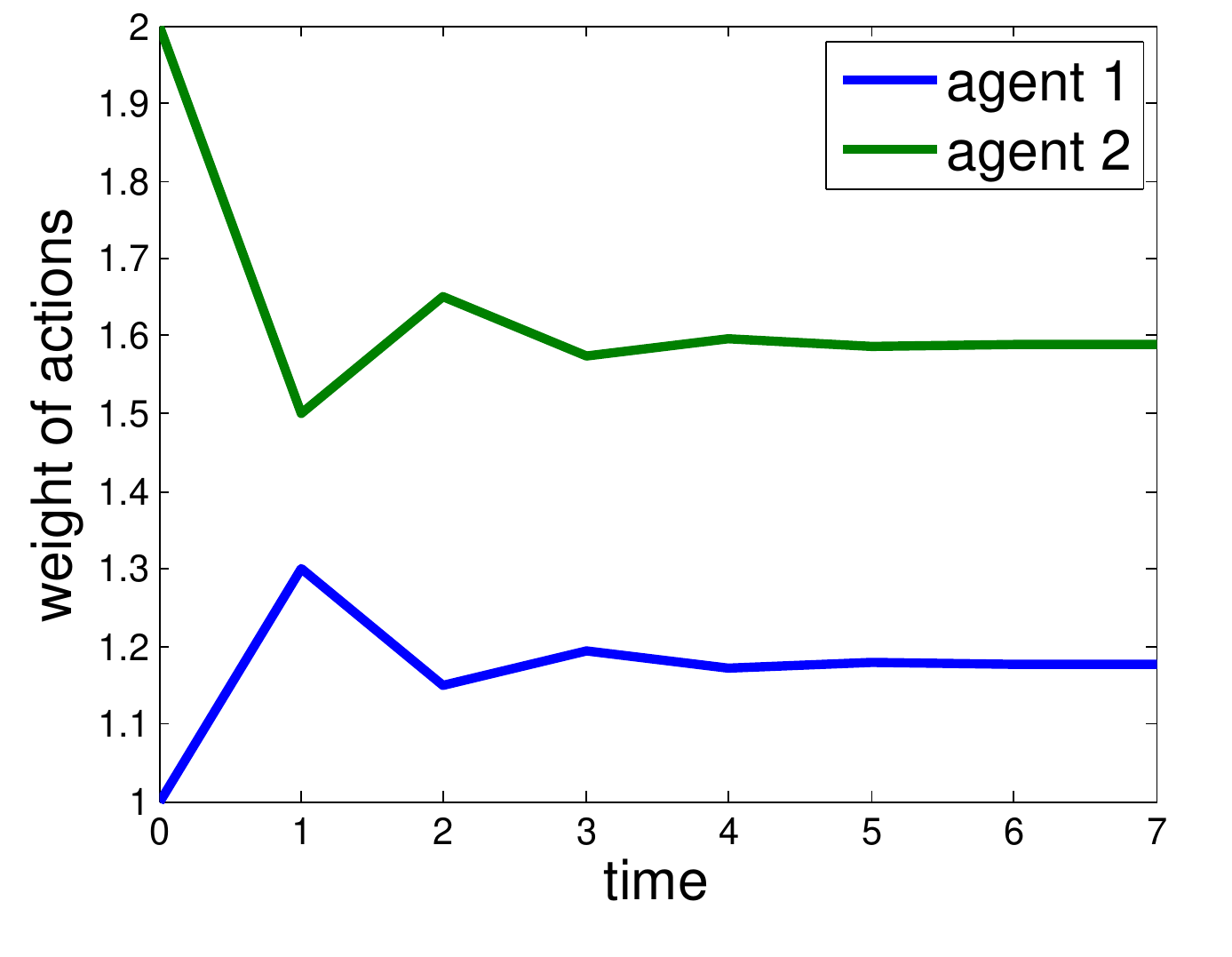}
\includegraphics[trim = 35mm 80mm 40mm 90mm, clip, width=0.5\textwidth, height=1.5in]{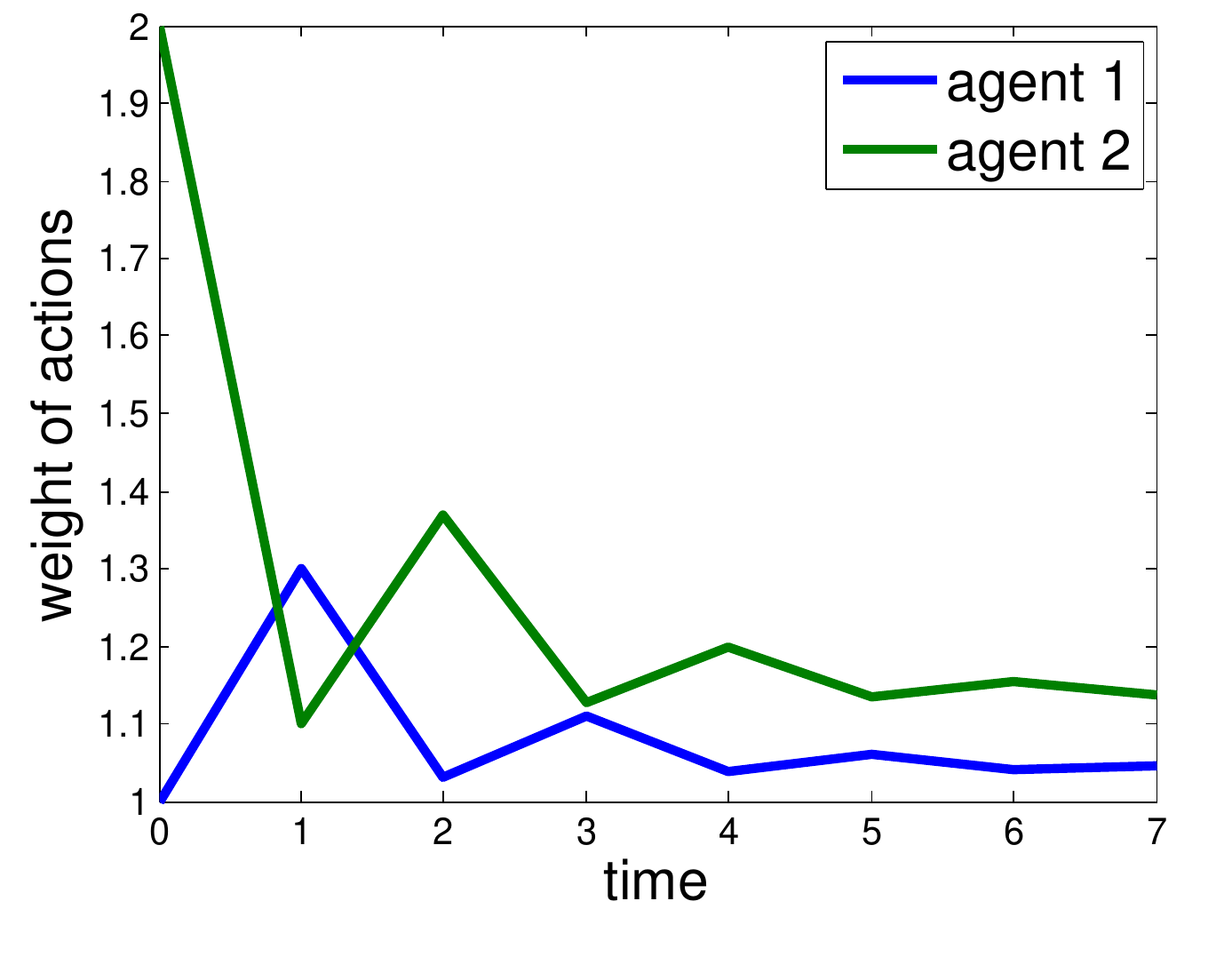}
}
\protect\caption{Simulation of actions for the synchronous case, with $r_1 + r_2 < 1$,
$r_2 = 0.5$ on the left, and $r_1 + r_2 > 1$, $r_2 = 0.9$ on the right.
This is a \name{fixed} - \name{fixed} reciprocation, with $k_1 = 1, k_2 = 2, r_1 = 0.3$.
Each agent's oscillate, while converging to her own limit.
}
\label{fig:fixed_fixed_03_05_09}%
\end{minipage}%
\hspace{0.04\textwidth}%
\begin{minipage}{.47\textwidth}
\centering
\includegraphics[trim = 25mm 20mm 80mm 210mm, clip, width=0.98\textwidth]{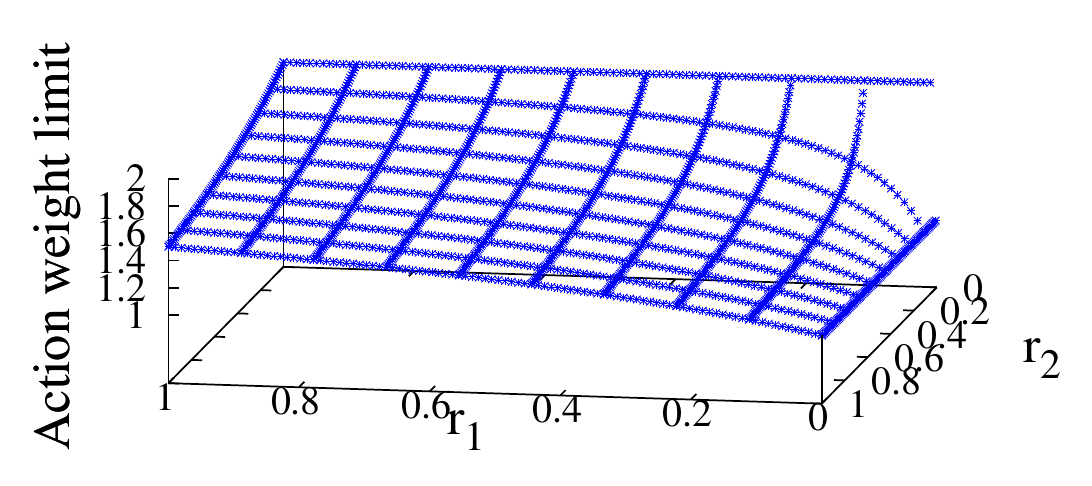}

\protect\caption{The common limit of the actions as a function of the reciprocity coefficients,
for a \name{Floating} - \name{floating} reciprocation, $k_1 = 1, k_2 = 2$.
Given $r_2$, agent $1$ receives most when $r_1 = 1$,
and given $r_1$, agent $2$ receives most when $r_2 = 0$.}
\label{fig:float_float_limit_agent12}%
\end{minipage}
\end{figure}

\subsection{\name{Floating} Reciprocation}

If both agents have the \name{floating} reciprocation attitude, their action sequences converge to a common limit.
\begin{theorem}\label{The:float_recip}
If %
\ifthenelse{\equal{\destin}{IJCAI16}}{
}{the reciprocation coefficients are neither both $0$ and nor both $1$, which means }%
$0 < r_1 + r_2 < 2$, then, as $t \to \infty$,
$x(t)$ and $y(t)$ converge to the same limit.
In the synchronous case ($T_1 = T_2 = T$), they both approach
\begin{equation*}
\frac{1}{2} \left(k_1 + k_2 + (k_2 - k_1)\frac{r_1 - r_2}{r_1 + r_2}\right)
= \frac{r_2}{r_1 + r_2} k_1 + \frac{r_1}{r_1 + r_2} k_2.
\end{equation*}
\end{theorem}
The common limit of the actions is shown
in~\figref{fig:float_float_limit_agent12}.

In Example~\ref{ex:colleages}, if agents $1$ and $2$ employ \name{fixed}
reciprocation, $r_1 = r_2 = 0.5, r_1' = r_2' = 0.0$ and
$k_1 = 0, k_2 = 0.5$, then we obtain
$L_x = L_y = (1/2) \cdot 0 + (1/2) \cdot 0.5 = 0.25$.

\ifthenelse{\equal{\destin}{IJCAI16}}{
The idea of the proof is to show that
$\set{[\min\set{x(t), y(t)}, \max\set{x(t), y(t)}]}_{t = 1}^\infty$ is
a nested sequence of segments, which lengths approach zero, and therefore,
$\set{x(t)}$ and $\set{y(t)}$ converge to the same limit. Finding this
limit stems from finding $\lim_{t \to \infty}\paren{x(t) + y(t)}$.
}{
Throughout the paper, whenever we need
concrete $T_1, T_2$, we consider the synchronous case.%
\ifthenelse{\equal{\destin}{IJCAI16}}{
The alternative case is omitted due to lack of space.
}{
The alternative case is fully handled in~\sectnref{Sec:altern}.
}%
\begin{proof}
We first prove that the convergence takes place. 

If both agents act at time $t > 0$, then $y(t) - x (t)$
\begin{align}
&\quad= x(t - 1) (r_2 - 1 + r_1) + y(t - 1) (1 - r_2 - r_1)\nonumber\\
\ifthenelse{\NOT \equal{\destin}{IJCAI16}}{
&\quad= y(t - 1) (1 - r_1 - r_2) - x(t - 1) (1 - r_1 - r_2)\nonumber\\
}{
}%
&\quad= (1 - r_1 - r_2) (y(t - 1) - x(t - 1)).%
\ifthenelse{\NOT \equal{\destin}{IJCAI16}}{
}{
\nonumber
}%
\label{eq:recur_diff_eq_both}
\end{align}
Since $0 < r_1 + r_2 < 2$, we have $\abs{(1 - r_1 - r_2)} < 1$.

If only agent $1$ acts at time $t > 0$, then $y(t) - x(t)$
\begin{eqnarray}
= y(t - 1) (1 - r_1) - x(t - 1) (1 - r_1)\\
= (1 - r_1) (y(t - 1) - x(t - 1)).%
\ifthenelse{\NOT \equal{\destin}{IJCAI16}}{
}{
\nonumber
}%
\label{eq:recur_diff_eq_one}
\end{eqnarray}
If $r_1 > 0$, then $\abs{(1 - r_1)} < 1$. 
Similarly, if only agent $2$ acts, then $y(t) - x (t) = (1 - r_2) (y(t - 1) - x(t - 1))$.
Since $r_1 + r_2 > 0$, either
$r_1$ or $r_2$ is greater than $0$, and since each agent acts an infinite number of
times,
we obtain
$\lim_{t \to \infty} \abs{y(t) - x(t)} = 0$. Since
$\forall t > 0: x(t), y(t) \in [\min\set{x(t - 1), y(t - 1)}, \max\set{x(t - 1), y(t - 1)}]$,
we have a nested sequence of segments, which lengths approach zero,
thus $x(t)$ and $y(t)$ both converge, and to the same limit.

Assume $T_1 = T_2 = T$ now, to find the common limit.
For all $t > 0$,
\begin{eqnarray*}
x(t) + y(t) = x(t - 1)(1 - r_1 + r_2) + y(t - 1) (r_1 + 1 - r_2)\nonumber\\
= x(t - 1) + y(t - 1) + (r_1 - r_2) (y(t - 1) - x(t - 1))\nonumber\\
\Rightarrow \lim_{t \to \infty}{x(y) + y(t)} = k_1 + k_2
+ \sum_{t = 0}^{\infty}{(r_1 - r_2) (y(t) - x(t))}\\%
\ifthenelse{\NOT \equal{\destin}{IJCAI16}}{
\stackrel{\eqnsref{eq:recur_diff_eq_both}}{=} k_1 + k_2
+ (r_1 - r_2)\sum_{t = 0}^{\infty}{(1 - r_1 - r_2)^t (k_2 - k_1)}\\
}{
}%
\underbrace{\stackrel{\text{geom. series}}{\to}}_{t \to \infty} k_1 + k_2
+ (r_1 - r_2)\frac{k_2 - k_1}{r_1 + r_2}\\
= k_1 + k_2
+ (k_2 - k_1)\frac{r_1 - r_2}{r_1 + r_2}.
\end{eqnarray*}

Since we have shown that both limits exist and are equal, each is equal
to half of $k_1 + k_2
+ (k_2 - k_1)\frac{r_1 - r_2}{r_1 + r_2}$.
\end{proof}

\ifthenelse{\NOT \equal{\destin}{IJCAI16}}{
If, unlike the theorem assumes, $r_1 + r_2 = 0$, then
since $r_1 + r_2 = 0 \iff r_1 = r_2 = 0$,
each agent keeps doing the same thing all the time: agent $1$ does $k_1$
and $2$ does $k_2$.

If, unlike the theorem assumes, $r_1 + r_2 = 2$, then
since $r_1 + r_2 = 2 \iff r_1 = r_2 = 1$,
each agent just repeats what the other one did last time.
In the synchronous case this means that they interchangeably play
$k_1$ and $k_2$.
}{
}%

The following gives the relation between $x$s and $y$s.
\begin{proposition}
If $r_1 + r_2 \leq 1$, then,
for every $t \geq 0: y(t) \geq x(t)$.
If $r_1 + r_2 \geq 1$, then,
$y(0) \geq x(0)$.
For every $t>0, t \in T_1 \cap T_2$, we have
$y(t - 1) \geq x(t - 1) \Rightarrow y(t) \leq x(t)$,
and $y(t - 1) \leq x(t - 1) \Rightarrow y(t) \geq x(t)$.
For any other $t$, we have
$y(t - 1) \geq x(t - 1) \Rightarrow y(t) \geq x(t)$,
and $y(t - 1) \leq x(t - 1) \Rightarrow y(t) \leq x(t)$.
In words, $x$s and $y$s alter their relative positions if and only if
both act.
\end{proposition}
\begin{proof}
Consider the case $r_1 + r_2 \leq 1$ first.
We employ induction. The basis is $t = 0$, where
$y(0) = k_2 \geq k_1 = x(0)$.

For the induction step, assume the proposition for all the times smaller than
$t > 0$ and prove it for $t$.
If only $1$ acts at $t$, then $y(t) = y(t - 1)$ and
$x(t) = (1 - r_1) x(t - 1) + r_1 y(t - 1)$. Therefore,
$y(t) \geq x(t)
\iff y(t - 1) \geq (1 - r_1) x(t - 1) + r_1 y(t - 1)$, which is equivalent to
$(1 - r_1) y(t - 1) \geq (1 - r_1) x(t - 1)$,
which holds by the induction hypothesis.
\ifthenelse{\NOT \equal{\destin}{IJCAI16}}{
If only agent $2$ acts at $t$, then
$x(t) = x(t - 1)$ and
$y(t) = (1 - r_2) y(t - 1) + r_2 x(t - 1)$. Therefore,
$y(t) \geq x(t)\\
\iff (1 - r_2) y(t - 1) + r_2 x(t - 1) \geq x(t - 1)\\
\iff (1 - r_2) y(t - 1) \geq (1 - r_2) x(t - 1),$
which is true by the induction hypothesis.
}{
The case where only $2$ acts at $t$ is similar.
}%

If both agents act at $t$, then
$x(t) = (1 - r_1) x(t - 1) + r_1 y(t - 1)$ and
$y(t) = (1 - r_2) y(t - 1) + r_2 x(t - 1)$. Therefore,
$y(t) \geq x(t)
\iff (1 - r_2) y(t - 1) + r_2 x(t - 1) \geq (1 - r_1) x(t - 1) + r_1 y(t - 1)
\iff (1 - r_1 - r_2) y(t - 1) \geq (1 - r_1 - r_2) x(t - 1),$
which is true by the induction hypothesis and using the assumption
$r_1 + r_2 \leq 1$.

Consider the case $r_1 + r_2 \geq 1$ now.
We employ induction again. The basis is $t = 0$, where
$y(0) = k_2 \geq k_1 = x(0)$.

For the induction step, assume the proposition for all values smaller than
$t > 0$ and prove it for $t$.
The cases where only agent $1$ acts at $t$
and where only $2$ acts at $t$ are shown analogously to how they are
shown for the case $r_1 + r_2 \leq 1$.
If both agents act at $t$, then we have shown that
$y(t) \geq x(t)\\
\iff (1 - r_1 - r_2) y(t - 1) \geq (1 - r_1 - r_2) x(t - 1),$
which means that $y(t - 1) \geq x(t - 1) \Rightarrow y(t) \leq x(t)$
and $y(t - 1) \leq x(t - 1) \Rightarrow y(t) \geq x(t)$,
assuming $r_1 + r_2 \geq 1$.
\end{proof}

The proposition implies that if $r_1 + r_2 \leq 1$, then $\set{x(t)}$ do not decrease
and $\set{y(t)}$ do not increase, since the next $x(t)$ (or $y(t)$) is either the same of a combination of the
last one with a higher value (lower value, for $y(t)$).

For $r_1 + r_2 > 1$, both $\set{x(t)}$ and $\set{y(t)}$ are not monotonic, unless
$T_1 \cap T_2 = \set{0}$, in which case they are monotonic, for the
reason above (in this case we always have $y(t) \geq x(t)$).
For $T_1 \cap T_2 \neq \set{0}$, take any positive $t$ in $T_1 \cap T_2$.
Then the larger value at $t - 1$ becomes the smaller one at $t$, thereby
getting smaller, and the smaller value gets larger analogously. In the
future, the smaller will only grow and the larger will decrease, thereby
behaving non-monotonically.
In particular, in the alternating case, each agent's actions alternate.
{This discussion assumes $r_1 < 1, r_2 < 1$,
to avoid getting $x(t) = y(t)$ when a single player acts.}
Some examples are simulated
in Figure~\ref{fig:float_float_03_05_09}.
\begin{figure}
\centering

\centerline{%
\includegraphics[width=0.24\textwidth, height=1.6in]{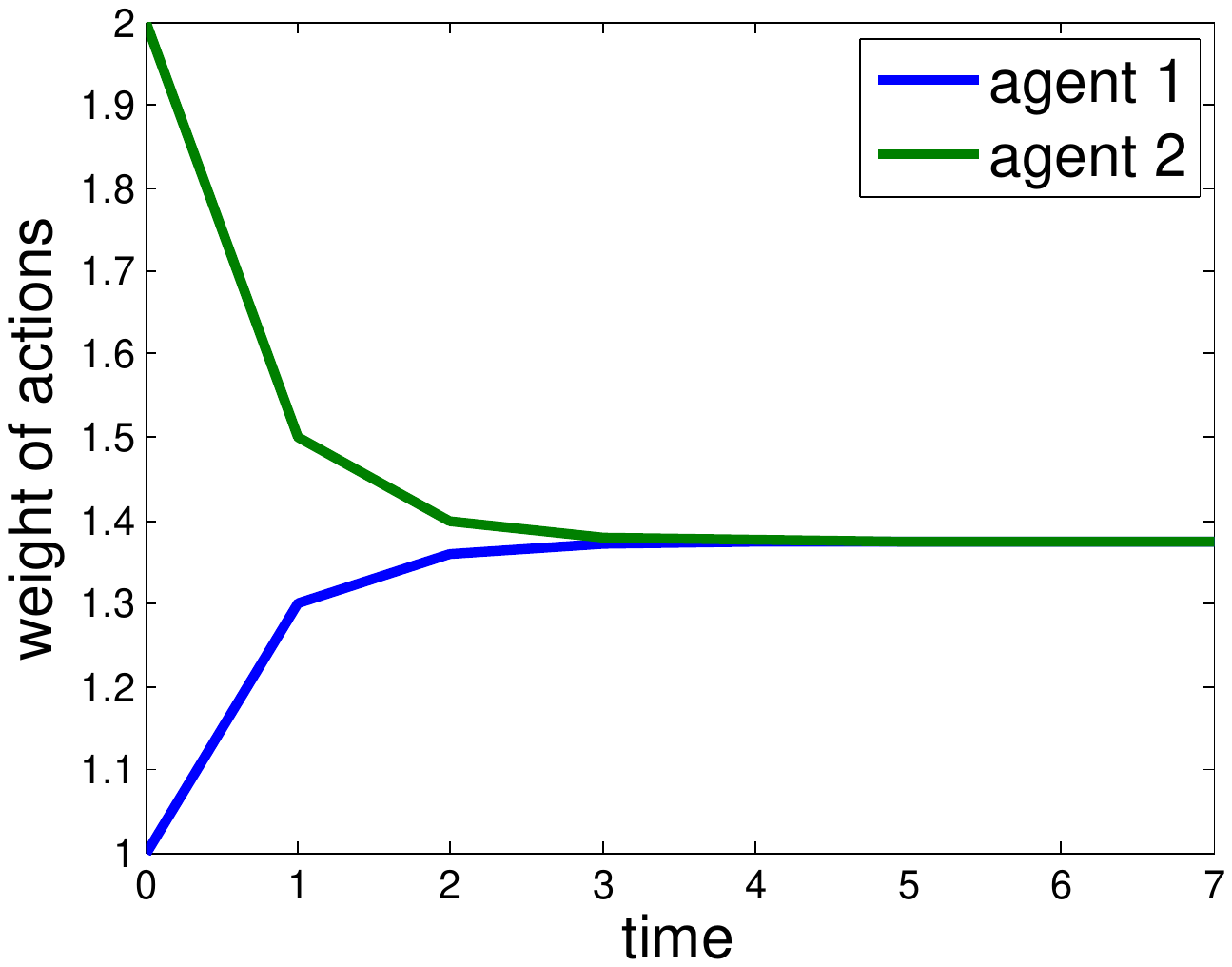}
\includegraphics[width=0.24\textwidth, height=1.6in]{float_float_03_09}
}%
\protect\caption{Simulation results for the synchronous case, $r_1 + r_2 < 1$
(left) and $r_1 + r_2 > 1$ (right). We assume
a \name{Floating} - \name{floating} reciprocation, $k_1 = 1, k_2 = 2, r_1 = 0.3$.
In the left graph, $r_2 = 0.5$, while in the right one, $r_2 = 0.9$.
In the left graph,
agent $1$'s actions are smaller than those of $2$;
agent $1$'s actions increase, while those of agent $2$ decrease.
In the right graph,
the actions of the agents alter their relative positions at each time step;
each agent's actions go up and down.
}

\label{fig:float_float_03_05_09}
\end{figure}
}%

\subsection{\name{Fixed} and \name{Floating} Reciprocation}\label{Sec:dynam_interact_pair:fix_float_recip}

Assume that agent $1$ employs the \name{fixed} reciprocation attitude,
while $2$ acts by the \name{floating} reciprocation.
We can show Theorem~\ref{The:fixed_float_recip} 
using the following lemma.
\begin{lemma}\label{lemma:fixed_float_recip:monot}
If $r_2 > 0$ and $r_1 + r_2 \leq 1$, then,
for every $t \geq t_{1, 1}: x(t + 1) \leq x(t)$, and for every
$t \geq 0: y(t + 1) \leq y(t)$.
\end{lemma}%
\ifthenelse{\NOT \equal{\destin}{IJCAI16}}{
We now prove this lemma.
\begin{proof}
We employ induction.
The basis consists of the following subcases: $t = 0, 0 < t < t_{1, 1}$
and $t = t_{1, 1}$.
For $t = 0$, we have either $y(1) = y(0)$ or $y(1) = (1 - r_2) k_2 + r_2 k_1
\leq k_2 = y(0)$.

For any $0 < t < t_{1, 1}$, we have either $y(t + 1) = y(t)$ or $y(t + 1) = (1 - r_2) y(t) + r_2 k_1
\stackrel{k_1 \leq y(t)}{\leq} (1 - r_2) y(t) + r_2 y(t) = y(t)$.

For $t = t_{1, 1}$, we either have $x(t_{1, 1} + 1) = x(t_{1, 1})$ or $x(t_{1, 1} + 1) = (1 - r_1) k_1 + r_1 y(t_{1, 1})$,
and anyway $x(t_{1, 1}) = (1 - r_1) k_1 + r_1 y(t_{1, 1} - 1)$ by the definition of $t = t_{1, 1}$.
Since $y(t_{1, 1}) \leq y(t_{1, 1} - 1)$ by the induction hypothesis, we have $x(t_{1, 1} + 1) \leq x(t_{1, 1})$.
As to $y$s, we have either $y(t_{1, 1} + 1) = y(t_{1, 1})$ or
\begin{eqnarray*}
y(t_{1, 1} + 1) = (1 - r_2) y(t_{1, 1}) + r_2 x(t_{1, 1})\\
\Rightarrow y(t_{1, 1} + 1) \leq y(t_{1, 1})
\stackrel{r_2 > 0}{\iff} x(t_{1, 1}) \leq y(t_{1, 1}).
\end{eqnarray*}
Either $y(t_{1, 1}) = y(t_{1, 1} - 1)$ or $y(t_{1, 1}) = (1 - r_2) y(t_{1, 1} - 1) + r_2 k_1$.
In the first case,
\begin{eqnarray*}
x(t_{1, 1}) \leq y(t_{1, 1})
\iff (1 - r_1) k_1 + r_1 y(t_{1, 1} - 1) \leq y(t_{1, 1} - 1)\\
\iff (1 - r_1) k_1 \leq (1 - r_1) y(t_{1, 1} - 1),
\end{eqnarray*}
which always holds.
In the second case,
\begin{eqnarray*}
x(t_{1, 1}) \leq y(t_{1, 1})\\
\iff (1 - r_1) k_1 + r_1 y(t_{1, 1} - 1) \leq (1 - r_2) y(t_{1, 1} - 1) + r_2 k_1\\
\iff (1 - r_1 - r_2) k_1 \leq (1 - r_1 - r_2) y(t_{1, 1} - 1),
\end{eqnarray*}
which is true, since $k_1 \leq y(t_{1, 1} - 1)$ and $(r_1 + r_2) \leq 1$.
Thus, the basis is proven.

For the induction step, for any $t > t_{1, 1}$, assume that
$ x(t) \leq x(t - 1) \leq \ldots \leq x(t_{1, 1})$, and
$y(t) \leq y(t - 1) \leq \ldots \leq y(0)$.

We now prove the lemma for $t$.
If $x(t + 1) = x(t)$, then trivially $x(t + 1) \leq x(t)$. Otherwise,
$x(t + 1) = (1 - r_1) k_1 + r_1 y(t) \leq (1 - r_1) k_1 + r_1 y(s_1(t) - 1) = x(t)$,
where the inequality stems from the induction hypothesis.
For $y$s, if $y(t + 1) = y(t)$, then trivially $y(t + 1) \leq y(t)$. Otherwise,
$y(t + 1) = (1 - r_2) y(t) + r_2 x(t) \leq (1 - r_2) y(s_2(t) - 1) + r_2 x(s_2(t) - 1) = y(t)$.
The above inequality stems from the induction hypothesis, if
$s_2(t) - 1 \geq t_{1, 1}$, so that the induction hypothesis for $x$s holds
as well as the one for $y$s.
Otherwise ($s_2(t) \leq t_{1, 1}$), the above inequality is proven as follows:
\begin{eqnarray*}
(1 - r_2) y(t) + r_2 x(t)\\ \leq (1 - r_2) y(s_2(t) - 1) + r_2 x(s_2(t) - 1)\\
\iff\\ (1 - r_2) (y(t) - y(s_2(t) - 1)) \leq r_2 (x(s_2(t) - 1) - x(t))\\ = r_2 (k_1 - x(t))
\end{eqnarray*}
Notice that $(y(t) - y(s_2(t) - 1)) = y(s_2(t)) - y(s_2(t) - 1)
\stackrel{s_2(t) \leq t_{1, 1}}{=} (1 - r_2) y(s_2(t) - 1) + r_2 k_1 - y(s_2(t) - 1)
= r_2 (k_1 - y(s_2(t) - 1))$. Thus,
\begin{eqnarray*}
(1 - r_2) (y(t) - y(s_2(t) - 1)) \leq r_2 (k_1 - x(t))\\
\iff (1 - r_2) r_2 (k_1 - y(s_2(t) - 1)) \leq r_2 (k_1 - x(t))\\
\stackrel{r_2 > 0}{\iff} (1 - r_2) (k_1 - y(s_2(t) - 1)) \leq (k_1 - x(t))\\
\iff x(t) - r_2 k_1 \leq (1 - r_2) y(s_2(t) - 1).
\end{eqnarray*}
To show this, notice that
$s_2(t) \leq t_{1, 1} < t \Rightarrow s_2(t) + 1 \leq t$.
In addition, $2$ acts at time slot $t_{1, 1} - 1$, and therefore
$s_2(t) \geq t_{1, 1} - 1$. Therefore, using the induction hypothesis
for $x$s we obtain
\begin{eqnarray*}
x(t) \stackrel{t \geq s_2(t) + 1 \geq t_{1, 1}} \leq x(s_2(t) + 1)
= (1 - r_1) k_1 + r_1 y(s_2(t)),
\end{eqnarray*}
where the equality stems from the fact that if $s_2(t) + 1 \notin T_1$,
then it is in $T_2$, and therefore $t = s_2(t)$, a contradiction.
Therefore,
\begin{eqnarray*}
x(t) - r_2 k_1 \leq (1 - r_1 - r_2) k_1 + r_1 y(s_2(t))\\
\leq (1 - r_1 - r_2) y(s_2(t)) + r_1 y(s_2(t))\\ = (1 - r_2) y(s_2(t))
\leq (1 - r_2) y(s_2(t) - 1),
\end{eqnarray*}
and the step has been proven.
\end{proof}
}{
The proof is by induction on $t$, using the definitions of reciprocation.
}%
With this lemma, we can prove the following.
\begin{theorem}\label{The:fixed_float_recip}
If $r_2 > 0$ and $r_1 + r_2 \leq 1$, then,
$\lim_{t \to \infty}{x(y)} = \lim_{t \to \infty}{y(t)} = k_1$.
\end{theorem}
\begin{proof}
We first prove that the convergence takes place, and then find its limit.
For each agent, Lemma~\ref{lemma:fixed_float_recip:monot} implies
that her actions are monotonically non-increasing.
Since the actions are bounded below by $k_1$, which can be easily
proven by induction, they both converge.

To find the limits, notice that in the limit we have
\begin{eqnarray}
(1 - r_1) k_1 + r_1 L_y = L_x\label{eq:recur_x_eq}\\
(1 - r_2) L_y + r_2 L_x = L_y.\label{eq:recur_y_eq}
\end{eqnarray}
From \eqnsref{eq:recur_y_eq}, we conclude that $L_x = L_y$, since $r_2 > 0$.
Substituting this to \eqnsref{eq:recur_x_eq} gives us
$L_x = L_y = k_1$, since $r_2 > 0$ and $r_1 + r_2 \leq 1$ imply $r_1 < 1$.
\end{proof}

\ifthenelse{\NOT \equal{\destin}{IJCAI16}}{

If, unlike the theorem assumes, $r_2 = 0$, then
agent $2$ keeps doing the same thing all the time: $k_2$, and
agent $1$ keeps doing $(1 - r_1) k_1 + r_1 k_2$ all the time when $t > 0$.

If, unlike the theorem assumes, $r_1 + r_2 > 1$, but the rest holds, then
it is still open what happens.
}{
}%

The relation between the sequences of $x$s and $y$s is given
by the following proposition (also covering the case $r_1 + r_2 \geq 1$).
\begin{proposition}\label{The:fix_float_recip:struct}
If $r_1 + r_2 \leq 1$, then for every $t \geq 0: y(t) \geq x(t)$.
If $r_1 + r_2 \geq 1$, then
$y(0) \geq x(0)$.
For every $t > 0$ such that $t \in T_1 \cap T_2$, we have
$y(t - 1) \leq x(t - 1) \Rightarrow y(t) \geq x(t)$.
For any $t \in T_1 \setminus T_2$, we have
$y(t) \geq x(t)$,
and for any $t \in T_2 \setminus T_1$, we have
$y(t - 1) \geq x(t - 1) \Rightarrow y(t) \geq x(t)$,
and $y(t - 1) \leq x(t - 1) \Rightarrow y(t) \leq x(t)$.
\end{proposition}%
\ifthenelse{\NOT \equal{\destin}{IJCAI16}}{
\begin{proof}
Consider the case $r_1 + r_2 \leq 1$ first.
We employ induction. The basis is $t = 0$, where
$y(0) = k_2 \geq k_1 = x(0)$.

For the induction step, we assume the proposition for all values smaller than
$t > 0$ and prove the proposition for $t$.
If only $1$ acts at $t$, then $y(t) = y(t - 1)$ and
$x(t) = (1 - r_1) k_1 + r_1 y(t - 1)$. Therefore,
\begin{eqnarray*}
y(t) \geq x(t)
\iff y(t - 1) \geq (1 - r_1) k_1 + r_1 y(t - 1)\\
\iff (1 - r_1) y(t - 1) \geq (1 - r_1) k_1,
\end{eqnarray*}
which is true.

If only agent $2$ acts at $t$, then $x(t) = x(t - 1)$ and
$y(t) = (1 - r_2) y(t - 1) + r_2 x(t - 1)$. Therefore,
\begin{eqnarray*}
y(t) \geq x(t)\\
\iff (1 - r_2) y(t - 1) + r_2 x(t - 1) \geq x(t - 1)\\
\iff (1 - r_2) y(t - 1) \geq (1 - r_2) x(t - 1),
\end{eqnarray*}
which is true by the induction hypothesis.

If both agents act at $t$, then
$x(t) = (1 - r_1) k_1 + r_1 y(t - 1)$ and
$y(t) = (1 - r_2) y(t - 1) + r_2 x(t - 1)$. Therefore,
\begin{eqnarray*}
y(t) \geq x(t)
\iff (1 - r_2) y(t - 1) + r_2 x(t - 1)\\ \geq (1 - r_1) k_1 + r_1 y(t - 1)\\
\iff (1 - r_1 - r_2) y(t - 1) \geq (1 - r_1) k_1 - r_2 x(t - 1).
\end{eqnarray*}
Since $x(t - 1) \geq k_1$, it is enough to show that
$(1 - r_1 - r_2) y(t - 1) \geq (1 - r_1) x(t - 1) - r_2 x(t - 1) = (1 - r_1 - r_2) x(t - 1)$,
which is true by the induction hypothesis and using the assumption
$r_1 + r_2 \leq 1$.
Thus, the case $r_1 + r_2 \leq 1$ has been proven.

Consider the case $r_1 + r_2 \geq 1$ now.
We employ induction. The basis is $t = 0$, where
$y(0) = k_2 \geq k_1 = x(0)$.

For the induction step, we assume the proposition for all values smaller than
$t > 0$ and prove the proposition for $t$.
The cases where only agent $1$ acts at $t$
and where only $2$ acts at $t$ are shown by analogy to how they are
shown for the case $r_1 + r_2 \geq 1$.
If both agents act at $t$, then we have shown that
\begin{eqnarray}
y(t) \geq x(t)\nonumber\\
\iff (1 - r_1 - r_2) y(t - 1) \geq (1 - r_1) k_1 - r_2 x(t - 1).
\label{eq:Lemma:fix_float_recip:struct:both_act}
\end{eqnarray}
Now, if $y(t - 1) \leq x(t - 1)$, then
$(1 - r_1 - r_2) y(t - 1) \geq (1 - r_1 - r_2) x(t - 1) \geq (1 - r_1) k_1 - r_2 x(t - 1)$,
and from \eqnsref{eq:Lemma:fix_float_recip:struct:both_act} we have $y(t) \geq x(t)$.
%
\end{proof}
}{
The proof employs induction on $t$.
}%

We note that although we have not seen yet whether
Theorem~\ref{The:fixed_float_recip} holds for $r_1 + r_2 > 1$, we know
that neither monotonicity (Lemma~\ref{lemma:fixed_float_recip:monot})
nor $y(t)$ being always at least as large as $x(t)$ or the other way
around holds in this case.
As a counterexample for both of them, consider the case of
$r_2 = 1, 0 < r_1 < 1, k_2 > k_1$.
One can readily prove by induction that for all $t$ we have
$x(2t + 1) >x(2t) = x(2t + 2)$ and $y(2t) > y(2t - 1) = y(2t + 1)$,
and thus both sequences are not monotonic.
In addition, one can inductively prove that
$x(2t + 1) > y(2t + 1), x(2t) < y(2t)$, and therefore no sequence is
always larger than the other one.

Figure~\ref{fig:fixed_float_03_05_09} shows how the actions evolve over time.
The actions seem to converge also in the unproven case $r_1 + r_2 > 1$.
\begin{figure}

\centerline{%
\includegraphics[trim = 35mm 80mm 40mm 90mm, clip, width=0.24\textwidth, height=1.5in]{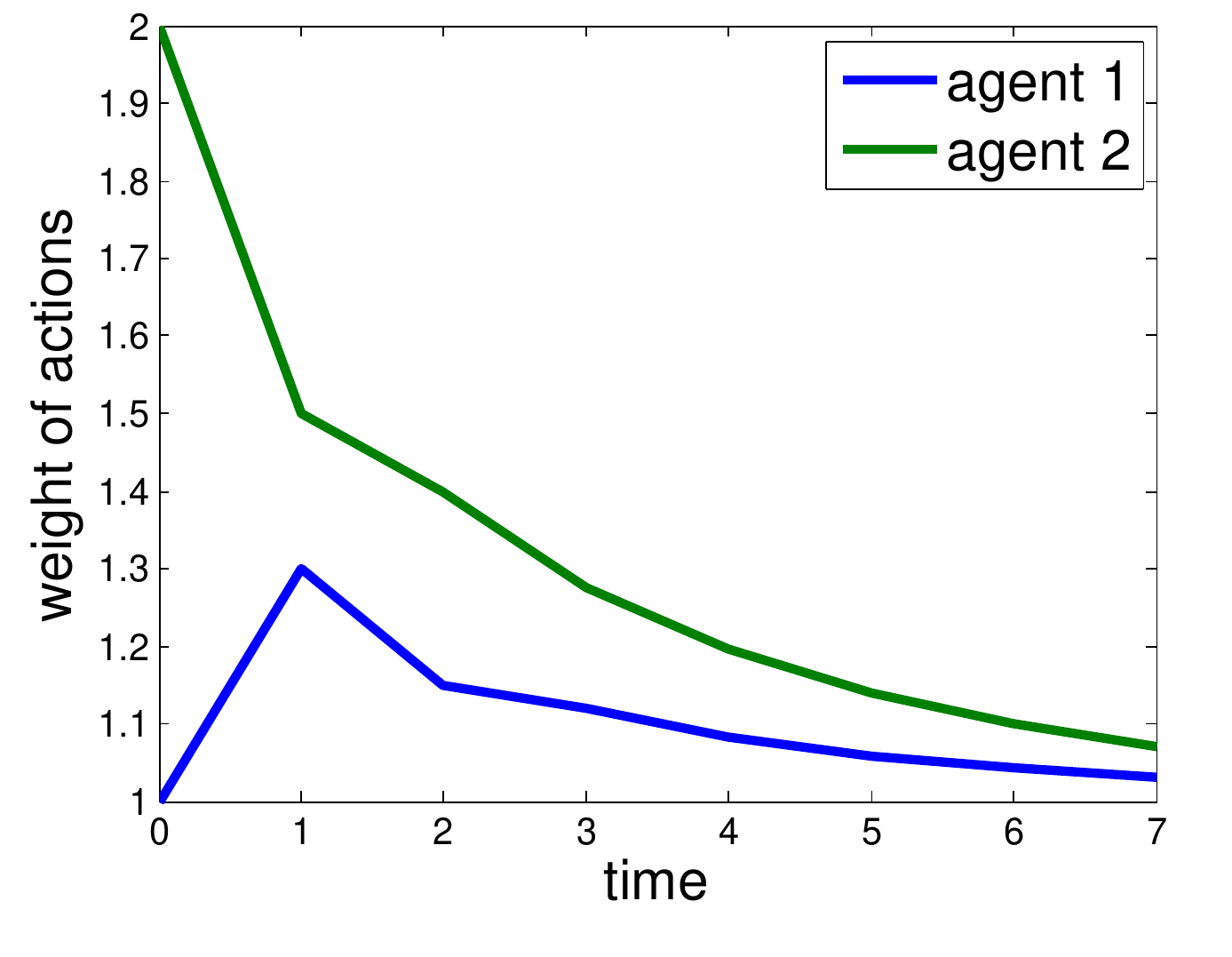}
\includegraphics[trim = 35mm 80mm 40mm 90mm, clip, width=0.24\textwidth, height=1.5in]{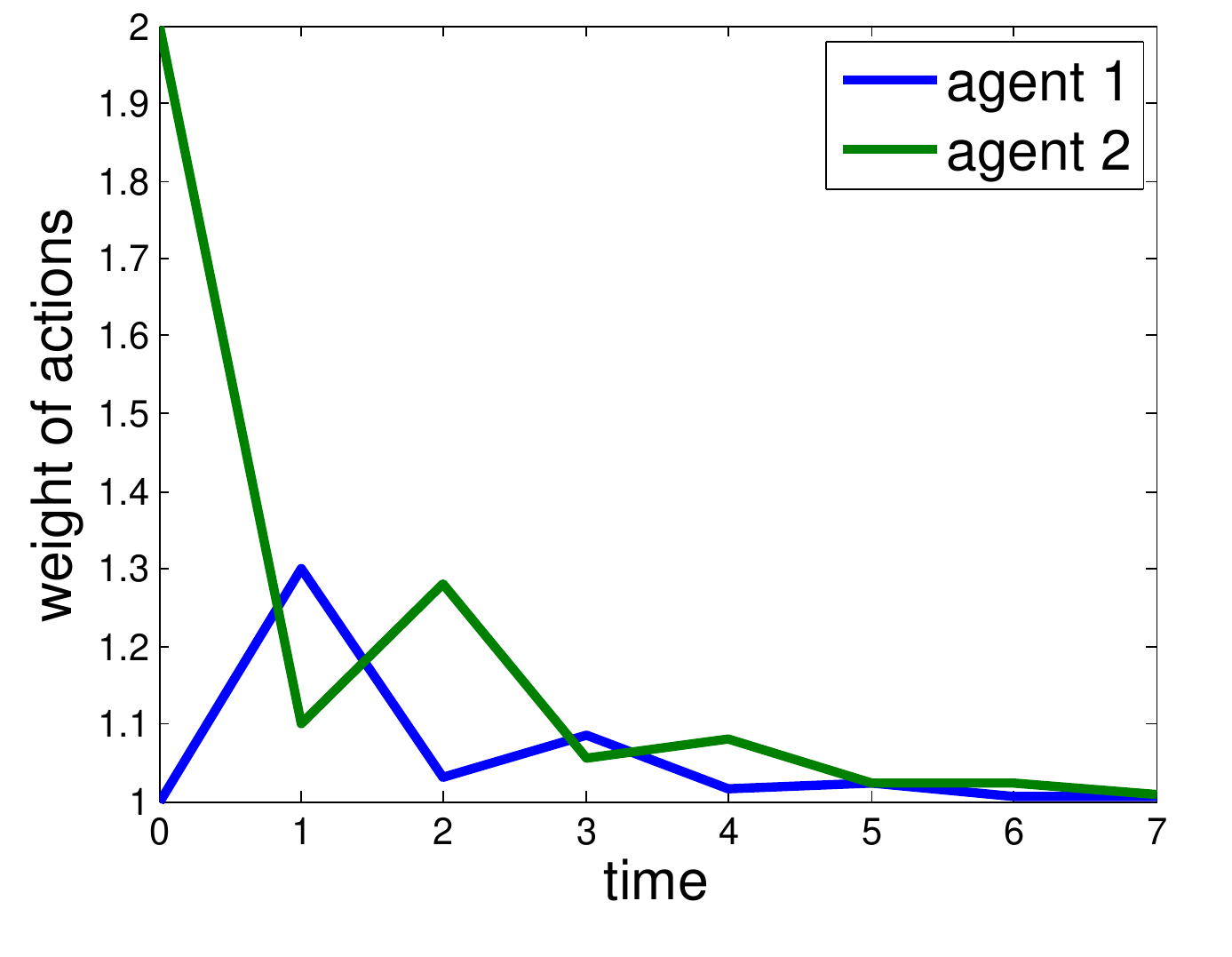}
}%
\protect\caption{Simulation of actions for the synchronous case, with $r_1 + r_2 < 1$,
$r_2 = 0.5$ on the left, and $r_1 + r_2 > 1$, $r_2 = 0.9$ on the right.
This is a \name{fixed} - \name{floating} reciprocation, with $k_1 = 1, k_2 = 2, r_1 = 0.3$.
In the left graph,
agent $1$'s actions are smaller than those of $2$;
agent $1$'s actions decrease after $t = 1$, while those of agent $2$ decrease all the time.
The common limits' value fits the theorem's prediction.
}

\label{fig:fixed_float_03_05_09}
\end{figure}

In the case of the mirroring assumption
that agent $1$ acts according to the \name{floating} reciprocation attitude,
while $2$ acts according to the \name{fixed} reciprocation, we can%
\ifthenelse{\equal{\destin}{IJCAI16}}{
obtain similar results, which are omitted due to lack of space.
}{
obtain the following similar results by analogy. 

\begin{theorem}\label{The:float_fixed_recip}
If $r_1 > 0$ and $r_1 + r_2 \leq 1$, then,
$\lim_{t \to \infty}{x(t)} =\lim_{t \to \infty}{y(t)} = k_2$.
\end{theorem}%
\ifthenelse{\NOT \equal{\destin}{IJCAI16}}{
The proof is analogous, with the lemma being
\begin{lemma}
If $r_1 > 0$ and $r_1 + r_2 \leq 1$, then,
for every $t \geq t_{2, 1}: y(t + 1) \geq y(t)$, and for every
$t \geq 0: x(t + 1) \geq x(t)$.
\end{lemma}


If, unlike the theorem assumes, $r_1 = 0$, then
agent $1$ keeps doing the same thing all the time: $k_1$, and
agent $2$ keeps doing $(1 - r_2) k_2 + r_2 k_1$ all the time when $t > 0$.

If, unlike the theorem assumes, $r_1 + r_2 > 1$, but the rest holds, then
it is still open what happens.
}{
}%

Regarding the relation between
$x$s and $y$s, we prove the following, by analogy to how Proposition~%
\ref{The:fix_float_recip:struct} is proven:
\begin{proposition}
If $r_1 + r_2 \leq 1$, then for every $t \geq 0: y(t) \geq x(t)$.
If $r_1 + r_2 \geq 1$, then,
$y(0) \geq x(0)$.
For every $t > 0$, such that $t \in T_1 \cap T_2$, we have
$y(t - 1) \leq x(t - 1) \Rightarrow y(t) \geq x(t)$.
For any $t \in T_2 \setminus T_1$, we have
$y(t) \geq x(t)$,
and for any $t \in T_1 \setminus T_2$, we have
$y(t - 1) \geq x(t - 1) \Rightarrow y(t) \geq x(t)$ and
$y(t - 1) \leq x(t - 1) \Rightarrow y(t) \leq x(t)$.
\end{proposition}
}%

\ifthenelse{\equal{\destin}{IJCAI16}}{
}{
Using Proposition~\ref{prop:gen_convergence_async_mixed}, we will prove
the following corollary.
\begin{corollary}\label{cor:fixed_float_recip_sync_gen}
Consider pairwise interaction, where agent $i$ is
\name{fixed}, and the other agent $j$ is
\name{floating}, and every agent acts at least once
every $r$ times.
Assume that $r_i < 1$ and $r_j > 0$.
Then, both actions sequences converge geometrically to $k_i$.
\end{corollary}
We omit the proof at this stage.
}%

For all the considered cases, we have the following
\begin{proposition}\label{prop:kind_mono_limit}
If both $L_x$ and $L_y$ exist, then $L_x \leq L_y$.
\end{proposition}
\ifthenelse{\NOT \equal{\destin}{AAMAS16}}{
\section{Alternating Case}\label{Sec:altern}

We consider the interaction of two agents, $1$ and $2$.
Some of the statements in the paper refer only to the synchronous case
($T_1 = T_2 = T$).
All of them can be updated for the alternating case
($T_1$ contains precisely
all the even times and $T_2$ contains zero and all the odd ones)%
\ifthenelse{\NOT \equal{\destin}{IJCAI16}}{,
and we show this in this section.
}{.
}%

Theorem~\ref{The:float_recip} can be extended as follows:
\begin{theorem}\label{The:float_recip_altern}
In the case where agents act alternately, which is when $T_1$ contains precisely
all the even times and $T_2$ contains zero and all the odd ones, they both approach
\begin{eqnarray*}
\frac{1}{2} \left(k_1 + k_2
+ \frac{(r_1 - r_2 - r_1 r_2)}{r_1 + r_2 - r_1 r_2} (k_2 - k_1) \right)\\
= \frac{r_2}{r_1 + r_2 - r_1 r_2} k_1 + \frac{r_1 - r_1 r_2}{r_1 + r_2 - r_1 r_2} k_2.
\end{eqnarray*}
\end{theorem}%
\ifthenelse{\NOT \equal{\destin}{IJCAI16}}{
\begin{proof}
We now assume the alternating case.
Consider the behavior of $x(t) + y(t)$.
For an even $t > 0$, only agent $1$ acts and we have
\begin{eqnarray*}
x(t) + y(t) = x(t - 1) + y(t - 1) + (r_1) (y(t - 1) - x(t - 1)).
\end{eqnarray*}
For an odd $t$, only $2$ acts and we have
\begin{eqnarray*}
x(t) + y(t) = x(t - 1) + y(t - 1) + (- r_2) (y(t - 1) - x(t - 1)).
\end{eqnarray*}
And therefore, we have
\begin{eqnarray*}
\Rightarrow \lim_{t \to \infty}{x(y) + y(t)} = k_1 + k_2
+ \sum_{t = 0}^{\infty}{(- r_2) (y(2t) - x(2t))}\\
+ \sum_{t = 0}^{\infty}{(r_1) (y(2t + 1) - x(2t + 1))}\\
\stackrel{\eqnsref{eq:recur_diff_eq_one}}{=} k_1 + k_2
- (r_2)\cdot \sum_{t = 0}^{\infty}{((1 - r_1)^{t} (1 - r_2)^t) (k_2 - k_1)}\\
+ (r_1)\cdot \sum_{t = 0}^{\infty}{((1 - r_1)^t (1 - r_2)^{t + 1}) (k_2 - k_1)}\\
= k_1 + k_2\\
+ \sum_{t = 0}^{\infty}{(1 - r_1)^t (1 - r_2)^t (r_1 - r_2 - r_1 r_2)) (k_2 - k_1)}\\
\underbrace{\stackrel{\text{geom. series}}{\to}}_{t \to \infty} k_1 + k_2\\
+ (r_1 - r_2 - r_1 r_2)(k_2 - k_1)\inv{1 - (1 - r_1)(1 -r_2)}\\
= k_1 + k_2
+ \frac{(r_1 - r_2 - r_1 r_2)}{r_1 + r_2 - r_1 r_2} (k_2 - k_1).
\end{eqnarray*}

Since we have shown that both limits exist and are equal, each equals
to half of $k_1 + k_2
+ \frac{(r_1 - r_2 - r_1 r_2)}{r_1 + r_2 - r_1 r_2} (k_2 - k_1)$.
\end{proof}
}{
The idea of the proof is proving that $x(t) + y(t)$ approach
$k_1 + k_2
+ \frac{(r_1 - r_2 - r_1 r_2)}{r_1 + r_2 - r_1 r_2} (k_2 - k_1)$.
}%

\ifthenelse{\NOT \equal{\destin}{IJCAI16}}{
If, unlike the theorem assumes, $r_1 + r_2 = 2$, then
since $r_1 + r_2 = 2 \iff r_1 = r_2 = 1$,
in the alternating case, agent $2$ plays at time $1$ the strategy
of agent $1$ at time $0$, which is $k_1$, and since then, each player plays it.
}{
}%

%
\ifthenelse{\equal{\content}{GT} \OR \equal{\content}{Two_agents} \OR \equal{\content}{Mult_agents} \OR \equal{\content}{All}}{
Proposition~\ref{The:float_recip_coeff:max_util_pair} can be generalized to
the alternating case as well, such that the statement stays true for
both the synchronous and the alternating cases.%
\ifthenelse{\NOT \equal{\destin}{IJCAI16}}{
The proof for
the alternating case follows.
\begin{proof}
Let us prove for agent $1$. We first express $1$'s utility and then maximize it. Assume first that $0 < r_2 < 1$, and therefore
$0 < r_1 + r_2 < 2$. Therefore, from Theorem~\ref{The:float_recip_altern},
\begin{eqnarray*}
\lim_{p \to \infty}{v(t_{1, p})} = \lim_{p \to \infty}{w(t_{1, p})}\\
= \frac{1}{2} \left(k_1 + k_2
+ \frac{(r_1 - r_2 - r_1 r_2)}{r_1 + r_2 - r_1 r_2} (k_2 - k_1) \right)\\
\Rightarrow
u_{1} = (1 - \beta_1) \frac{1}{2} \left(k_1 + k_2
+ \frac{(r_1 - r_2 - r_1 r_2)}{r_1 + r_2 - r_1 r_2} (k_2 - k_1) \right).
\end{eqnarray*}

To find a maximum point of this utility as function of $r_1$, we
differentiate:
\begin{eqnarray*}
\frac{\partial (u_{1})}{\partial (r_1)}
= \ldots = \frac{(1 - \beta_1) (k_2 - k_1) r_2 - (r_2)^2}{(r_1 + r_2 - r_1 r_2)^2}
\end{eqnarray*}
Therefore, the derivative is zero either for all $r_1$ or for none. In any
case, the maximum is attained at an endpoint, so we check the utility at
those points now.

At $r_1 = 0$, it is $(1 - \beta_1) k_1$.
At $r_1 = 1$, it is $(1 - \beta_1) (r_2 k_1 + (1 - r_2) k_2)$. Since
$r_2 k_1 + (1 - r_2) k_2$ is a convex combination of $k_1$
and $k_2$, this is at least $k_1$. Thus, agent $1$ prefers $r_1 = 1$
if and only if $1 - \beta_1 \geq 0$, and we
have proven the proposition for agent $1$ choosing, when $r_2 > 0$.

If $r_2 = 0$, then we notice that $r_1 = r_2 = 0$ results in
$u_{1} = k_2 - \beta_1 k_1$, which is the largest possible utility,
so choosing $r_1 = 0$ is optimal here, and we
have proven the proposition for agent $1$ choosing.

The case of agent $2$ choosing $r_2$ is proven by analogy, besides
the case $r_1 = 0$. In this case, for $r_2 = 1$ agent $2$'s utility is
$k_1 - \beta_2 k_1$. Since for $r_1 = 0$ agent $1$ will make an action
of $k_1$ regardless $r_2$, this is the best possible utility for $2$ here,
since the performed action is minimized.
\end{proof}
}{
The idea of the proof is expressing the utility and maximizing
it by  differentiating. We obtain that the derivative is zero either
always or never, so we check the endpoints.
}%

\ifthenelse{\NOT \equal{\destin}{IJCAI16}}{
The discussion in \sectnref{Sec:sw_maxim:converg_opt} for
the \name{floating} reciprocation remains true, when we employ
the generalization of Proposition~\ref{The:float_recip_coeff:max_util_pair}
for the alternating case.
}{
}%
}{
}%

}{
}%
\section{Multi-Agent Interaction}\label{Sec:dynam_interact_interdep}%
\ifthenelse{\equal{\content}{Process} \OR \equal{\content}{GT} \OR \equal{\content}{All}}{

We now analyze the general interdependent interaction, when agents interact with many agents.
}{

We analyze reciprocation in the above model.
}%
%
To formally discuss the actions after the interaction
has settled down, we consider the limits (if exist)\footnote{Agent $i$ acts at the times in $T_i = \set{t_{i, 0} = 0, t_{i, 1}, t_{i, 2}, \ldots}$.}
$\lim_{p \to \infty}{\imp_{i, j}(t_{1, p})}$,
and $\lim_{t \to \infty}{x_{i, j}(t)}$, for agents $i$ and $j$.
Since the sequence $\set{x_{i,j}(t)}$ is $\set{\imp_{i, j}(t_{1, p})}$ with finite repetitions,
the limit $\lim_{p \to \infty}{\imp_{i, j}(t_{1, p})}$ exists if and only if $\lim_{t \to \infty}{x_{i, j}(t)}$ does.
If they exist, they are equal.
Denote $L_{i, j} \defas \lim_{t \to \infty}{x_{i, j}(t)}$.

\anote{One idea is to find a reduction to the the problem of two agents.
Another idea is to try to generalize the case of two agents.}

\ifthenelse{\NOT \equal{\destin}{IJCAI16}}{
\subsection{Convergence}
}{
}%

\ifthenelse{\equal{\destin}{IJCAI16}}{
}{
We show that in the synchronous case,
for every two agents $i, j$ such that $(i, j) \in E$, actions $x_{i, j}(t)$ converge
to a strictly positive combination of all the kindness values.
The rate of convergence is geometric.
}%

We first provide general convergence results, and then
we find the common limit for the case when at most one agent is
\name{fixed} and synchronous in Theorem~\ref{The:gen_convergence}.
We finally use simulations to analyze the limits in other cases.
In this section, the ambivalent case of $r_i + r_i' = 1$ is taken to be \name{floating}.

First, we have convergence for the case of \name{floating}
agents.
\begin{proposition}\label{prop:gen_convergence_async_float}
Consider a connected interaction graph,
where all agents are \name{floating} and for every agent $i$,
$r_i + r_i' < 1$.
Then,
for all pairs of agents $i \neq j$ such that $(i, j) \in E$, the limit $L_{i, j}$ exists;
all these limits
are equal to each other.
\end{proposition}
\begin{proof}
Follows directly
from \cite[Theorem~$2$]{BlondelHendrickxOlshevskyTsitsiklis2005}.
This article and similar articles on multiagent
coordination~\cite{Moreau2005,TsitsiklisBertsekasAthans1986}
prove convergence when all agents are
\name{floating}.
\end{proof}

We now show convergence, when
some agents are \name{fixed}.
\begin{proposition}\label{prop:gen_convergence_async_mixed}
Consider a connected interaction graph,
where for all agents $i$, $r_i' > 0$.
Assume that at least one agent employs the \name{fixed} attitude
and every agent acts at least once every $q$ times,
for a natural $q > 0$.
Then,
for all pairs of agents $i \neq j$ such that $(i, j) \in E$, the limit $L_{i, j}$ exists.
The convergence is geometrically fast.
\end{proposition}
\begin{proof}
We express 
how each action depends on the actions in the previous time
in matrix $A(t) \in \realsP^{\abs{E} \times \abs{E}}$,
which, in the synchronous case, is defined as follows:
\begin{equation}
A(t)((i, j), (k, l)) \defas
\begin{cases}
(1 - r_i - r_i') & \text{if } k = i, l = j;\\
r_i + r_i' \frac{1}{\abs{\outNeighb(i)}} & \text{if } k = j, l = i;\\
r_i' \frac{1}{\abs{\outNeighb(i)}} & \text{if } k \neq j, l = i;\\
0 & \text{otherwise},
\end{cases}
\label{eq:dynam_mat_def}
\end{equation}
where the first line is missing
	for the \name{fixed} agents,
since for them, own behavior does not matter.
If, for
each time $t \in T$, the column
vector $\vec{p(t)} \in \realsP^{\abs{E}}$ describes the actions at time $t$,
in the sense that its $(i, j)$th coordinate contains $x_{i, j}(t)$ (for $(i,j)\in E$), then
we have $\vec{p}(t + 1) = A(t + 1) \vec{p}(t) + \vec{k'}$,
where $\vec{k'}$ is the relevant kindness vector, formally defined as
\begin{equation*}
k'(t)((i, j)) \defas
\begin{cases}
(1 - r_i - r_i') k_i & \text{if } i \text{ is \name{fixed}};\\
0 & \text{otherwise}.
\end{cases}
\end{equation*}
In a not necessarily synchronous case,
only a subset of agents act at a given time~$t$. For an acting
agent $i$, every $A(t)((i, j), (k, l))$ is defined as in the synchronous
case.
For a non-acting agent~$i$, we define
\begin{equation}
A(t)((i, j), (k, l)) \defas
\begin{cases}
1 & \text{if } k = i, l = j;\\
0 & \text{otherwise}.
\end{cases}
\label{eq:dynam_mat_def_async}
\end{equation}
The kindness vector is defined as
\begin{equation*}
k'(t)((i, j)) \defas
\begin{cases}
(1 - r_i - r_i') k_i & \text{if } i \text{ is \name{fixed} and acting};\\
0 & \text{otherwise}.
\end{cases}
\end{equation*}
By induction, we obtain
$\vec{p}(t) = \prod_{t' = 1}^t{A(t')} \vec{p}(0) + \sum_{\vec{k'} \in K}{\set{ \paren{\sum_{l \in S_{\vec{k'}}(t)}\prod_{t' = l }^t{A(t')}} \vec{k'} }}$,
where $K$ is the set of all possible kindness vectors and $S_{\vec{k'}}(t)$
is a set of the appearance times of $\vec{k'}$, which are at most $t$.

We aim to show that $\vec{p(t)}$ converges.
First, defining $r_i(M)$ to be the sum of the $i$th row of $M$,
note \cite[\eqns{$3$}]{ButlerSiegel2013}, namely
\begin{eqnarray}
r_i(AB) = \sum_{j = 1}^n{\sum_{k = 1}^n{a_{i,j} b_{j, k}}}
= \sum_{j = 1}^n{a_{i, j} r_j(B)}.
\label{eq:sum_row_prod}
\end{eqnarray}
Since the sum of every row in any $A(t)$ is at most $1$, we conclude
that if $B \leq \beta C$, $C_{i, j} \equiv 1$, then also $A(t) B \leq \beta C$.

We now prove that an upper bound of the form $\beta C$ on the entries
of $\prod_{t = p}^q{A(t)}$ converges to zero geometrically.
We have just shown that this bound never increases.
First, $A(p) \leq C$, yielding the bound in the beginning.
Now, let $i$ be a \name{fixed} agent, and assume he acts at time $t$.
Thus, each row in $A(t)$ which relates to the edges entering $i$
sums to less than $1$, and from \eqnsref{eq:sum_row_prod} we gather that
the upper bound on the appropriate rows in $A(t) B$ decreases relatively
to the bound on $B$ by some constant ratio. Since the graph is connected,
for all agents $i$, $r_i' > 0$, and every agent acts every $q$ times, we
will have, after enough multiplications, that the bound on all the entries
will have decreased by a constant ratio.

Every agent acts at least once every $q$ times,
so we gather that for some $q' > 0$, every $q'$ times, the product
of matrices becomes at most a given fraction of the product $q'$ times
before. This implies a geometric convergence of $\prod_{t' = 1}^t{A(t')}$.
As for
$\sum_{l \in S_{\vec{k'}}(t)}\prod_{t' = l }^t{A(t')}$,
we have proven an exponential upper bound, thus
$\sum_{l \in S_{\vec{k'}}(t)}\prod_{t' = l }^t{A(t')}
\leq \sum_{l \in S_{\vec{k'}}(t)}{ \alpha^{\brackt{\frac{t - l + 1}{q'}}} C }
\leq \sum_{l \in S_{\vec{k'}}(t)}{ \alpha^{\frac{t - l}{q'}} C }
= \alpha^{\frac{t}{q'}} \paren{\sum_{l \in S_{\vec{k'}}(t)}{ \alpha^{\frac{- l}{q'}} }} C
\stackrel{\leq}{\text{\small geom. seq.}} \frac{\alpha^{\frac{t}{q'}} - 1}{\alpha^{\frac{1}{q'}} - 1} C$,
proving a geometric convergence of the series
$\sum_{l \in S_{\vec{k'}}(t)}\prod_{t' = l }^t{A(t')}$.
Therefore, $\vec{p(t)}$ converges, and it does so geometrically fast.
\end{proof}

As an immediate conclusion of this proposition, we can finally generalize%
\ifthenelse{\NOT\equal{\content}{Mult_agents}}{
Theorem~\ref{The:fixed_float_recip}
}{
The convergence theorem for pairwise interaction of a \name{fixed}
and a \name{floating} agent.
}%
to the case $r_1 + r_2 > 1$
as follows.
\begin{corollary}
Consider pairwise interaction, where one agent $i$ employs
\name{fixed} reciprocation and the other agent $j$ employs
the \name{floating} one,
and every agent acts at least once every $q$ times.
Assume that $0 < r_i < 1$ and $r_j > 0$.
Then, both limits exist and are equal to $k_i$.
The convergence is geometrically fast.
\end{corollary}
\begin{proof}
Proposition~\ref{prop:gen_convergence_async_mixed}
implies geometrically fast convergence.
We find the limits as in the proof of%
\ifthenelse{\NOT\equal{\content}{Mult_agents}}{
Theorem~\ref{The:fixed_float_recip}.
}{
the theorem for pairwise interaction that being generalized.
}%
\end{proof}
\ifthenelse{\NOT \equal{\destin}{IJCAI16}}{
We now consider two cases which fall beyond this corollary's assumptions,
while we do assume synchroneity, to simplify the details.

If, unlike the theorem assumes, $r_i = 1$, then, if $r_j < 1$, then
we know from the \name{float}-\name{float} case that both action
sequences approach $\frac{r_j}{r_1 + r_2} k_i + \frac{1}{r_1 + r_2} k_j$.
If $r_i = r_j = 1$, then both agents switch from $k_1$ to $k_2$.

If, unlike the theorem assumes, $r_j = 0$, then, agent
$i$ constantly acts $(1 - r_i) k_i + r_i k_j$ and $j$ constantly acts
$k_j$.
}{
}%

We now turn to finding the limit. We manage to do this only in the
synchronous case, when all
the agents are \name{floating} or all the \name{fixed} agents have the
same kindness. For all reciprocation attitudes, the following theorem also
provides an alternative proof of convergence in the synchronous case.
\begin{theorem}\label{The:gen_convergence}
Given a connected interaction graph,
consider the synchronous case where for all agents $i$, $r_i' > 0$.
If 
there exists a cycle of an odd length in the graph 
(or at least one agent $i$ employs \name{floating} reciprocation and
has $r_i + r_i' < 1$),
then,
for all pairs of agents $i \neq j$ such that $(i, j) \in E$, the limit $L_{i, j}$ exists
and it is a positive combination of all the kindness values of the agents
who are \name{fixed}, if at least one agent is \name{fixed}, and of all
the kindness values $k_1, \ldots, k_n$, if all agents are \name{floating}.
The convergence is geometrically fast.
Moreover, if all agents employ \name{floating} reciprocation, then all these limits
are equal to each other and it is a convex combination
of the kindness values, namely
\begin{eqnarray}
L = \frac{\sum_{i \in N}{\paren{\frac{d(i)}{r_i + r_i'} \cdot k_i}}}{\sum_{i \in N}{\paren{\frac{d(i)}{r_i + r_i'}}}}.\label{eq:all_float_lim}
\end{eqnarray}
If, on the other hand, all the \name{fixed} agents have the same
kindness $k$, then all these limits are equal to $k$.
In any case, when not all the agents are \name{floating}, then changing
only the kindness of the \name{floating} agents leaves all the limits
as before (also follows from the limits being positive combinations of
all the kindness values of the agents who are fixed).
\end{theorem}
Let us say several words about the assumptions.
If all agents are \name{fixed}, we can prove
that the actions are subsequences of the
actions in the synchronous case (a straightforward generalization of~%
\ifthenelse{\NOT\equal{\content}{Mult_agents}}{
Lemma~\ref{lemma:fixed_recip:subseq_synch_case}.)
}{
The appropriate lemma for pairwise interaction.
}%
Thus, the synchronous case represents all the cases
in the limit, when all agents are \name{fixed}.
The  assumption of a cycle of an odd length virtually always holds,
since three people influencing each other form such a cycle.

\begin{proof}
We first prove the case where all agents use \name{floating} reciprocation.
We express 
how each action depends on the actions in the previous time
in a matrix,
and prove the theorem by applying the famous Perron--Frobenius theorem~%
\cite[Theorem~$1.1$,~$1.2$]{seneta2006} to this matrix.
We now define the dynamics matrix $A \in \realsP^{\abs{E} \times \abs{E}}$:
\begin{equation}
A((i, j), (k, l)) \defas
\begin{cases}
(1 - r_i - r_i') & \text{if } k = i, l = j;\\
r_i + r_i' \frac{1}{\abs{\outNeighb(i)}} & \text{if } k = j, l = i;\\
r_i' \frac{1}{\abs{\outNeighb(i)}} & \text{if } k \neq j, l = i;\\
0 & \text{otherwise}.
\end{cases}
\label{eq:dynam_mat_def_float}
\end{equation}
According to the definition of \name{floating} reciprocation, if for
each time $t \in T$ the column
vector $\vec{p(t)} \in \realsP^{\abs{E}}$ describes the actions at time $t$,
in the sense that its $(i, j)$th coordinate contains $x_{i, j}(t)$ (for $(i,j)\in E$), then
$\vec p(t + 1) = A \vec p(t)$. We then call $\vec{p(t)}$
an action vector.
Initially, $\vec p_{(i, j)}(0) = k_i$.

Further, we shall need to use
the Perron--Frobenius theorem for primitive matrices.
We now prepare to use it, and first we show that $A$ is primitive.
First, $A$ is irreducible since we can move from any $(i, j) \in E$ to any
$(k, l) \in E$ as follows.
We can move from an action to its reverse, since
if $k = j, l = i$, then $A((i, j), (k, l)) = r_i + r_i' \frac{1}{\abs{\outNeighb(i)}} > 0$.
We can also move from an action to another action by the same agent, since
we can move to any action on the same agent and then to its reverse.
To move to an action on the same agent, notice that
if $l = i$, then
$A((i, j), (k, l)) \geq r_i' \frac{1}{\abs{\outNeighb(i)}} > 0$.
Now, we can move from any action $(i, j)$ to any other action $(k, l)$
by moving to the reverse action $(j, i)$ (if $k = j, l = i$, we are done).
Then, follow a path from $j$ to $k$ in graph $G$ by moving to
the appropriate action by an agent and then to the reverse, as many
times as needed till we are at
the action $(k, j)$ and finally to the action $(k, l)$.
Thus, $A$ is irreducible.

By definition, $A$ is non-negative. $A$ is
aperiodic, since either at least one agent $i$ has $r_i + r_i < 1$,
and thus the diagonal contains non-zero elements, or there exists a cycle
of an odd length in the interaction graph $G$.
In the latter case, let the cycle be $i_1, i_2, \ldots, i_p$ for an odd $p$. Consider
the following cycles on the index set of the matrix:
$(i, j), (j, i), (i, j)$ for any $(i, j) \in E$
and $(i_2, i_1), (i_3, i_2), \ldots, (i_p, i_{p - 1}), (i_1, i_p), (i_2, i_1)$. Their lengths are $2$ and $p$,
respectively, which greatest common divisor is $1$,
implying aperiodicity.
Being
irreducible and aperiodic, $A$ is primitive
by~\cite[Theorem~$1.4$]{seneta2006}.
Since the sum of every row is $1$, the spectral radius is $1$.

According to
the Perron--Frobenius theorem for primitive matrices~%
\cite[Theorem~$1.1$]{seneta2006},
the absolute values of all eigenvalues except one eigenvalue of $1$
are strictly less than $1$. The eigenvalue $1$ has unique right and left eigenvectors, up to a constant factor. Both these eigenvectors
are strictly positive.
Therefore,~\cite[Theorem~$1.2$]{seneta2006} implies that
$\lim_{t \to \infty}{A^t} = \vec{1} \vec{v}'$, where $\vec{v}'$ is
the left eigenvector of the value~$1$, normalized such that
$\vec{v}' \vec{1} = 1$,
and the approach rate is geometric.
Therefore, we obtain
$\lim_{t \to \infty}{\vec{p(t)}}
= \lim_{t \to \infty}{A^t \vec{p}(0)}
= \vec{1} \vec{v}' \vec{p}(0) = \vec{1} \sum_{(i, j) \in E}{v'((i, j)) k_i}$.
Thus, actions converge to $\vec{1}$ times
$\sum_{(i, j) \in E}{v'((i, j)) k_i}$.

To find this limit, consider the vector $v'$ defined by
$v'((i, j))  = \frac{1}{r_i + r_i'}$. Substitution shows it is a left
eigenvector of A. To normalize it such that
$\vec{v}' \vec{1} = 1$, divide this vector by the sum of its coordinates,
which is $\sum_{i \in N}{\frac{d(i)}{r_i + r_i'}}$, obtaining
$v'((i, j))  = \frac{1}{\sum_{i \in N}{\frac{d(i)}{r_i + r_i'}}} \cdot \frac{1}{r_i + r_i'}$.
Therefore, the common limit is
$\frac{\sum_{i \in N}{\paren{\frac{d(i)}{r_i + r_i'} \cdot k_i}}}{\sum_{i \in N}{\paren{\frac{d(i)}{r_i + r_i'}}}}$.

We now prove the case where at least one agent employs \name{fixed} reciprocation.
We define the dynamics matrix $A$ analogously to the previous case,
besides that the first line from~\eqnsref{eq:dynam_mat_def_float} is missing
	for the \name{fixed} agents,
since for them, own behavior does not matter.
In this case, we have $\vec{p}(t + 1) = A \vec{p}(t) + \vec{k'}$,
where $\vec{k'}$ is the relevant kindness vector, formally defined as
\begin{equation*}
k'((i, j)) \defas
\begin{cases}
(1 - r_i - r_i') k_i & \text{if } i \text{ is \name{fixed}};\\
0 & \text{otherwise}.
\end{cases}
\end{equation*}
By induction, we obtain
$\vec{p}(t) = A^t \vec{p}(0) + \paren{\sum_{l = 0}^{t - 1}{A^l}} \vec{k'}$.

Analogically to the previous case, $A$ is irreducible and non-negative.
As shown above, $A$ is aperiodic, and therefore,
primitive.
Since at least one agent employs \name{fixed} reciprocation, at least one row
of $A$ sums to less than $1$, and therefore the spectral radius of $A$
is strictly less than $1$.

Now, the Perron--Frobenius implies that all the eigenvalues are strictly
smaller than $1$. Since we have
$\lim_{t \to \infty}{\vec{p}(t)}
= \lim_{t \to \infty}{A^t \vec{p}(0)} + \paren{\lim_{t \to \infty}{\sum_{l = 0}^{t - 1}{A^l}}} \vec{k'}$,
~\cite[Theorem~$1.2$]{seneta2006} implies that this limits exist
(the first part converges to zero, while the second one is a series of
geometrically decreasing elements.) Since $A$ is primitive,
$\paren{\lim_{t \to \infty}{\sum_{l = 0}^{t - 1}{A^l}}} > 0$.

When all the fixed agents have the same kindness $k$, we now find the limits.
Taking the limits in the equality
$\vec{p}(t + 1) = A \vec{p}(t) + \vec{k'}$ yields
$(I - A) \lim_{t \to \infty}\vec{p}(t) = \vec{k'}$.
\cite[Lemma~$B.1$]{seneta2006} implies that $I - A$ is invertible and therefore, if we
guess a vector $\vec{x}$ that fulfills $(I - A) \vec{x} = \vec{k'}$,
it will be the limit. Since the vector with all actions equal to $k$
satisfies this equation, we conclude that all the limits are equal
to $k$.
In any case, when there exists at least one \name{fixed} agent, changing
only the kindness of the \name{floating} agents will not change
the (unique) solution of $(I - A) \vec{x} = \vec{k'}$, and, therefore,
will not change the limits.
\anote{
By the way, the series for the inverse of $I - A$ from Lemma~$B.1$ yields
once more that
$\lim_{t \to \infty}\vec{p}(t) = \paren{\lim_{t \to \infty}{\sum_{l = 0}^{t - 1}{A^l}}} \vec{k'}$.}
\end{proof}

Let us consider several examples of \eqnsref{eq:all_float_lim}.

\begin{example}
If the interaction graph is regular, meaning that all the degrees are
equal to each other, we have
$L = \frac{\sum_{i \in N}{\paren{\frac{k_i}{r_i + r_i'}}}}{\sum_{i \in N}{\paren{\frac{1}{r_i + r_i'}}}}$.
This holds for cliques, modeling small human collectives or groups of countries,
and for cycles, modeling circular computer networks.
\end{example}

\begin{example}
For star networks, modeling networks of
a supervisor of several people or entities, assume w.l.o.g.\ that
agent $1$ is the center, and we have
$L = \frac{\frac{n - 1}{r_1 + r_1'} \cdot k_1 + \sum_{i \in N \setminus \set{1}}{\paren{\frac{k_i}{r_i + r_i'} }}}{\frac{n - 1}{r_1 + r_1'} + \sum_{i \in N \setminus \set{1}}{\paren{\frac{1}{r_i + r_i'}}}}$.
\end{example}

An obvious conclusion of the theorem is that the \name{fixed} agents are,
intuitively spoken, more important than the \name{floating} ones, at least
their kindness is.
We now conclude about the optimal reciprocation, which goes back to
providing decision support.
\begin{proposition}\label{prop:float_max_L}
If \eqnsref{eq:all_float_lim} holds, then agent~$i$ who wants
to maximize the common value $L$, and who can choose
either $r_i$ or $r_i'$, in certain limits
$[a, b]$, for $a > 0$, should choose either the smallest possible
or the largest possible coefficient, as follows.
We assume we choose $r_i$, but the same holds for $r_i'$ with the obvious
adjustments.
She should set $r_i$ to $b$, if
$\sum_{j \in N \setminus \set{i}}{\paren{\frac{d(j)}{r_j + r_j'} \cdot k_j}}
-k_i \paren{\sum_{j \in N \setminus \set{i}}{\paren{\frac{d(j)}{r_j + r_j'}}}}$
is positive, to $a$, if that is negative, and to an arbitrary, if zero.
When this expression is not zero, only these choices are optimal.
\end{proposition}
\ifthenelse{\equal{\destin}{IJCAI16}}{
The  proof considers the sign of the derivative, and is omitted due
to lack of space.
}{
\begin{proof}
Consider the derivative:
\begin{eqnarray*}
\frac{\partial L}{\partial (r_i)}
= \frac{\frac{-d(i) k_i}{(r_i + r_i')^2} \paren{\sum_{j \in N}{\paren{\frac{d(j)}{r_j + r_j'}}}}
+ \sum_{j \in N}{\paren{\frac{d(j)}{r_j + r_j'} \cdot k_j}} \frac{d(i)}{(r_i + r_i')^2}}
{\paren{\sum_{j \in N}{\paren{\frac{d(j)}{r_j + r_j'}}}}^2}\\
= \frac{ \frac{d(i)}{(r_i + r_i')^2} \paren{
\sum_{j \in N}{\paren{\frac{d(j)}{r_j + r_j'} \cdot k_j}}
-k_i \paren{\sum_{j \in N}{\paren{\frac{d(j)}{r_j + r_j'}}}} }}
{\paren{\sum_{j \in N}{\paren{\frac{d(j)}{r_j + r_j'}}}}^2}\\
= \frac{ \frac{d(i)}{(r_i + r_i')^2} \paren{
\sum_{j \in N \setminus \set{i}}{\paren{\frac{d(j)}{r_j + r_j'} \cdot k_j}}
-k_i \paren{\sum_{j \in N \setminus \set{i}}{\paren{\frac{d(j)}{r_j + r_j'}}}} }}
{\paren{\sum_{j \in N}{\paren{\frac{d(j)}{r_j + r_j'}}}}^2}.
\end{eqnarray*}
Therefore, the derivative is zero either for all $r_i$ or for none. In any
case, the maximum is attained at an endpoint. To avoid complicated
substitution, we consider the derivative sign instead:
\begin{eqnarray*}
\frac{\partial u_i}{\partial r_i} \geq 0
\iff \sum_{j \in N \setminus \set{i}}{\paren{\frac{d(j)}{r_j + r_j'} \cdot k_j}}
-k_i \paren{\sum_{j \in N \setminus \set{i}}{\paren{\frac{d(j)}{r_j + r_j'}}}} \geq 0,
\end{eqnarray*}
and so when $\sum_{j \in N \setminus \set{i}}{\paren{\frac{d(j)}{r_j + r_j'} \cdot k_j}}
-k_i \paren{\sum_{j \in N \setminus \set{i}}{\paren{\frac{d(j)}{r_j + r_j'}}}}$
is nonnegative, $i$ should choose the largest $r_i$, which is $b$,
and she should choose $r_i = a$ otherwise.
When the derivative is not zero, these choices are the only optimal ones.
\end{proof}
}

\ifthenelse{\equal{\destin}{IJCAI16}}{
}{
We can also prove a general convergence result, allowing agents
to act in a more general way than modeled above.
We need the
following definition:
\begin{defn}
Given a metric space $(X, d)$,
function $f \colon X \to X$ is called \defined{contraction}, if for any
$x_1, x_2 \in X$, we have $d(f(x_1), f(x_2)) \leq q d(x_1, x_2)$,
for a $q \in (0, 1)$.
\end{defn}

\begin{theorem}\label{The:gen_convergence_contract}
Given an interaction graph,
assume the synchronous case, where every agent acts in the following way.
Let $S \subseteq \reals$ be a compact set.
As in the proof of Theorem~\ref{The:gen_convergence}, assume that
for
each time $t \in T$, the column
vector $\vec{p(t)} \in S^{\abs{E}}$ describes the actions at time $t$,
in the sense that its $(i, j)$th\footnote{For $(i, j) \in E$.} coordinate contains $x_{i, j}(t)$,
and that there exists a contraction $f \colon S^{\abs{E}} \to S^{\abs{E}}$
with respect to the Euclidean metric,
such that
$\vec p(t + 1) = f(\vec p(t))$.
Initially, $\vec p_{(i, j)}(0) = k_i$.
Then,
for all pairs of agents $i \neq j$ such that $(i, j) \in E$, the limit $L_{i, j}$ exists.
The convergence is geometrically fast.
\end{theorem}
This theorem is not a generalization of Theorem~\ref{The:gen_convergence},
since matrix in $A$ in the proof of Theorem~\ref{The:gen_convergence}
needs not be a contraction.
\begin{proof}
By definition of action, $\vec{p(t)} = f^t(\vec{p(0)})$,
and using Banach's fixed point theorem~\cite[Exercise~$6.88$]{hewitt1975},
we know that $f^t(\vec{p(0)})$ converges to the unique fixed point of $f$
in $S^{\abs{E}}$, with a geometrical speed, thereby proving the theorem.
\end{proof}
}%

\ifthenelse{\NOT \equal{\destin}{IJCAI16}}{
\subsection{The Process}\label{Sec:dynam_interact_interdep:process}
}{
}%

\section{Simulations}\label{Sec:dynam_interact_interdep_sim}

We now answer some theoretically unanswered questions from \sectnref{Sec:dynam_interact_interdep} using MatLab simulations,
running at least $100$ synchronous rounds, to achieve
practical convergence.

We first concentrate on the case of three agents who can influence
each other, meaning that the interaction graph is a clique.
We begin by corroborating the already proven result that when at least
one \name{fixed} agents exists, then the kindness of the \name{floating}
agents does not influence the actions in the limit. Another proven thing
we corroborate is that when exactly one \name{fixed} agent exists, then
all the actions approach her kindness as time approaches infinity.
When the actions are plotted as functions of time,
we obtain graphs such as those in \figref{fig:act_time_fi_fl_fl_A_1_2_3_3_1_5}.
The left graph on that figure demonstrates, that exponential convergence
may be quite slow, and this is a new observation we did not know from the theory.
We also corroborate that the limiting values of the actions depend
linearly on the kindness values of all the fixed agents, the
proportionality coefficients being independent of the other kindness
values. 
In order to reasonably cover the sampling space, all the above mentioned
regularities have also been automatically checked for the combinations of
kindness values of $1, 2, 3, 4, 5$, over $r_i$ and $r_i'$ values
of $0.1, 0.3, 0.5, 0.7, 0.9$ and over all the relevant reciprocation
attitudes. The checks were up to the absolute precision of $0.01$.

We do not know the exact limits when there exist two or more
\name{fixed} agents with distinct kindness values.
We at least know that the dependencies on the kindness values are linear,
but we lack theoretical knowledge about the dependencies of the limits of
actions on the reciprocation coefficients, so we simulate the interaction
for various reciprocation coefficients, obtaining graphs like those
in \figref{fig:act_lim_r_1}, and analogously for the dependency on $r_1'$. Note that
we can have both increasing and decreasing graphs in the same scenario,
and also convex and concave graphs.
The observed monotonicity was automatically verified for all the
above mentioned combinations of parameters.
This monotonicity means that if an agent wants to maximize the limit
of the actions of some agent on some other agent, she can do this
by choosing an extreme value of $r_i$ or $r_i'$.

\ifthenelse{\NOT \equal{\destin}{IJCAI16}}{
A natural question is whether Proposition~\ref{prop:kind_mono_limit}
can be extended for more than two agents. Since the kindness of
the \name{floating} agents does not effect the limits, this sort of
monotonicity with respect to kindness would require all the limits of
the actions to be the same. We therefore ask whether the monotonicity
holds at least for the actions of the \name{fixed} agents.
The answer is negative, as \figref{fig:act_time_fi_fi_fi_A_1_5_2} shows.
}{
}%

The next thing we study is a fourth agent, interacting with some
of the other agents. We consider the limits of the actions as functions
of the fourth agent's degree, but
we found no regularity in these graphs; in particular,
no monotonicity holds in the general case.


\begin{figure}[ht]
\centerline{%
\includegraphics[trim = 35mm 80mm 40mm 90mm, clip, width=0.24\textwidth, height=1.5in]{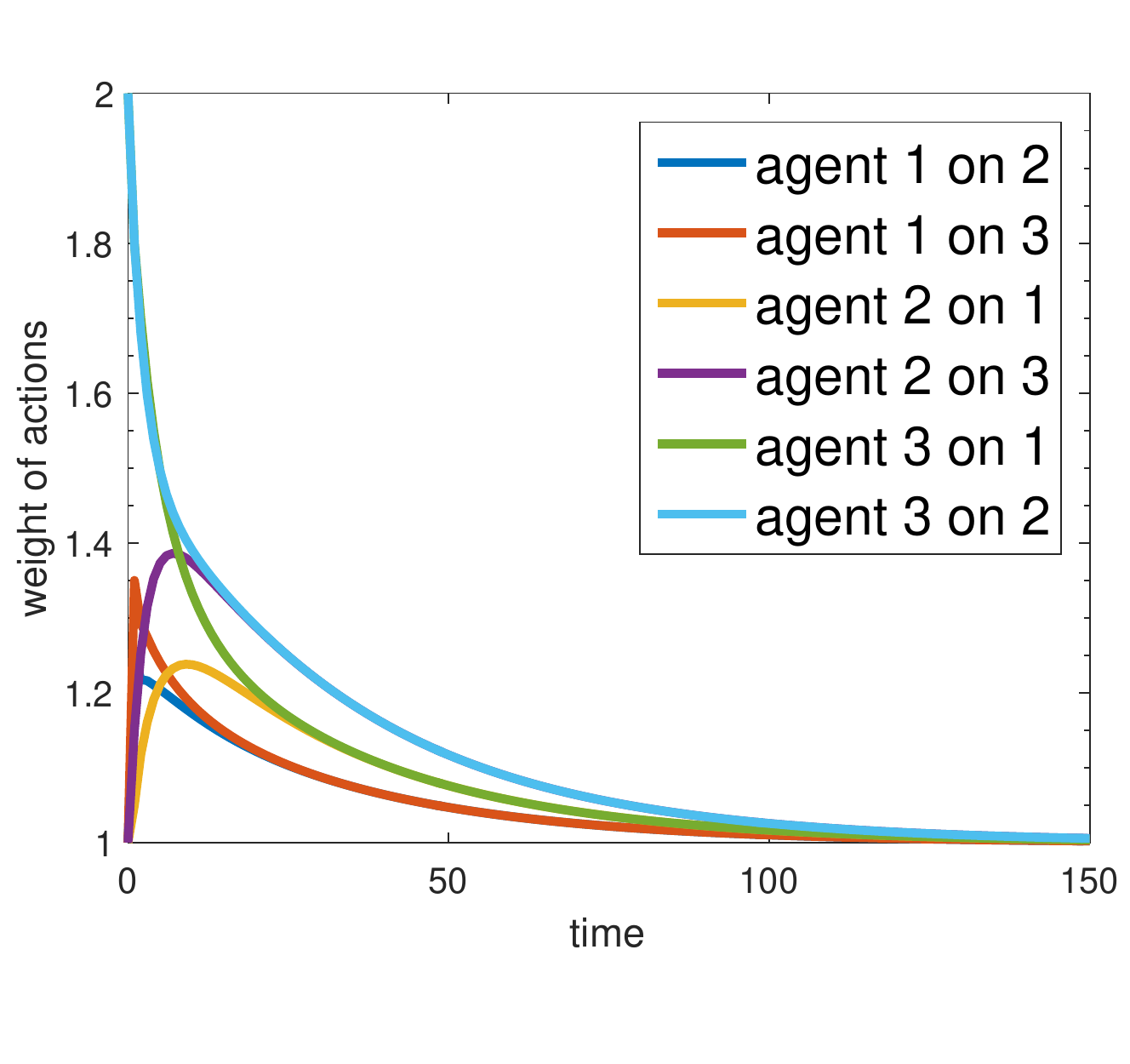}
\includegraphics[trim = 35mm 80mm 40mm 90mm, clip, width=0.24\textwidth, height=1.5in]{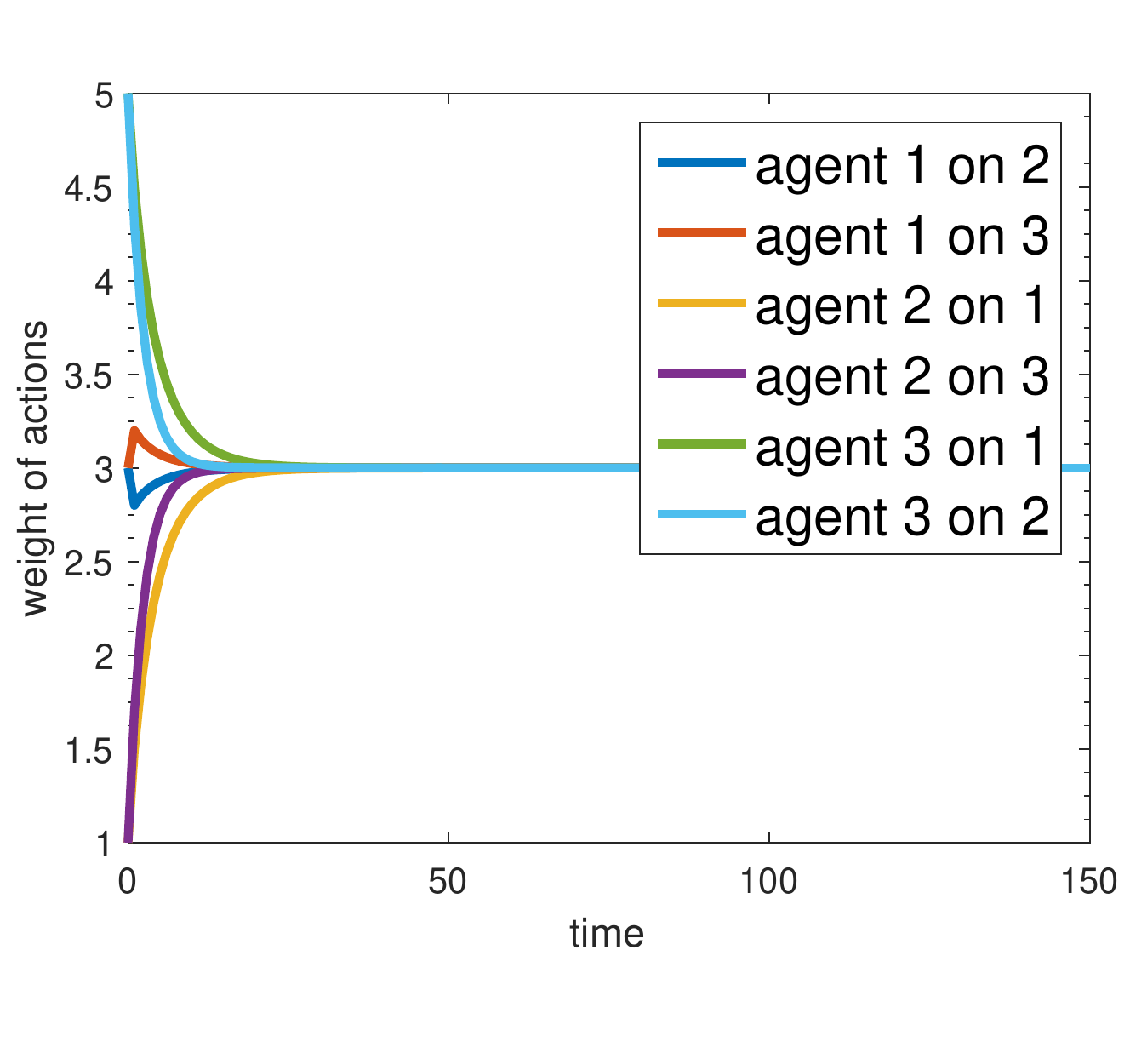}
}%
\protect\caption{Simulation results for the synchronous case, with one \name{fixed}
and two \name{floating} agents, for
$r_1 = 0.1, r_2 = 0.1, r_3 = 0.1, r_1' = 0.5, r_2' = 0.1, r_3' = 0.1$.
In the left graph, $k_1 = 1, k_2 = 1, k_3 = 2$, while in the right one, $k_1 = 3, k_2 = 1, k_3 = 5$.
The common limits, which are equal to the kindness of agent $1$, fit
the prediction of Theorem~\ref{The:gen_convergence}.
}
\label{fig:act_time_fi_fl_fl_A_1_2_3_3_1_5}
\end{figure}
%
%
\begin{figure}[ht]
\centerline{%
\includegraphics[trim = 35mm 80mm 40mm 90mm, clip, width=0.24\textwidth, height=1.5in]{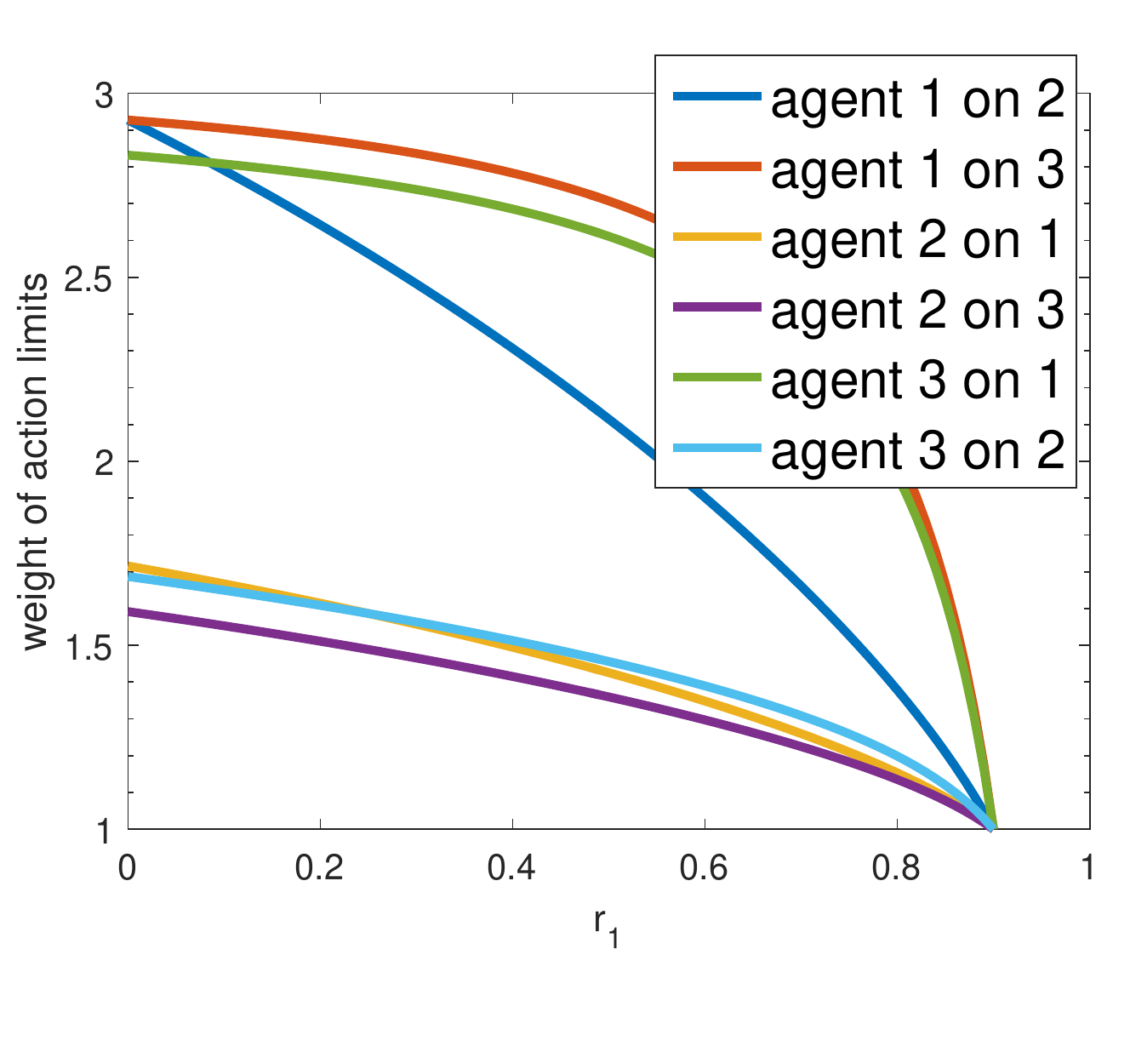}
\includegraphics[trim = 35mm 80mm 40mm 90mm, clip, width=0.24\textwidth, height=1.5in]{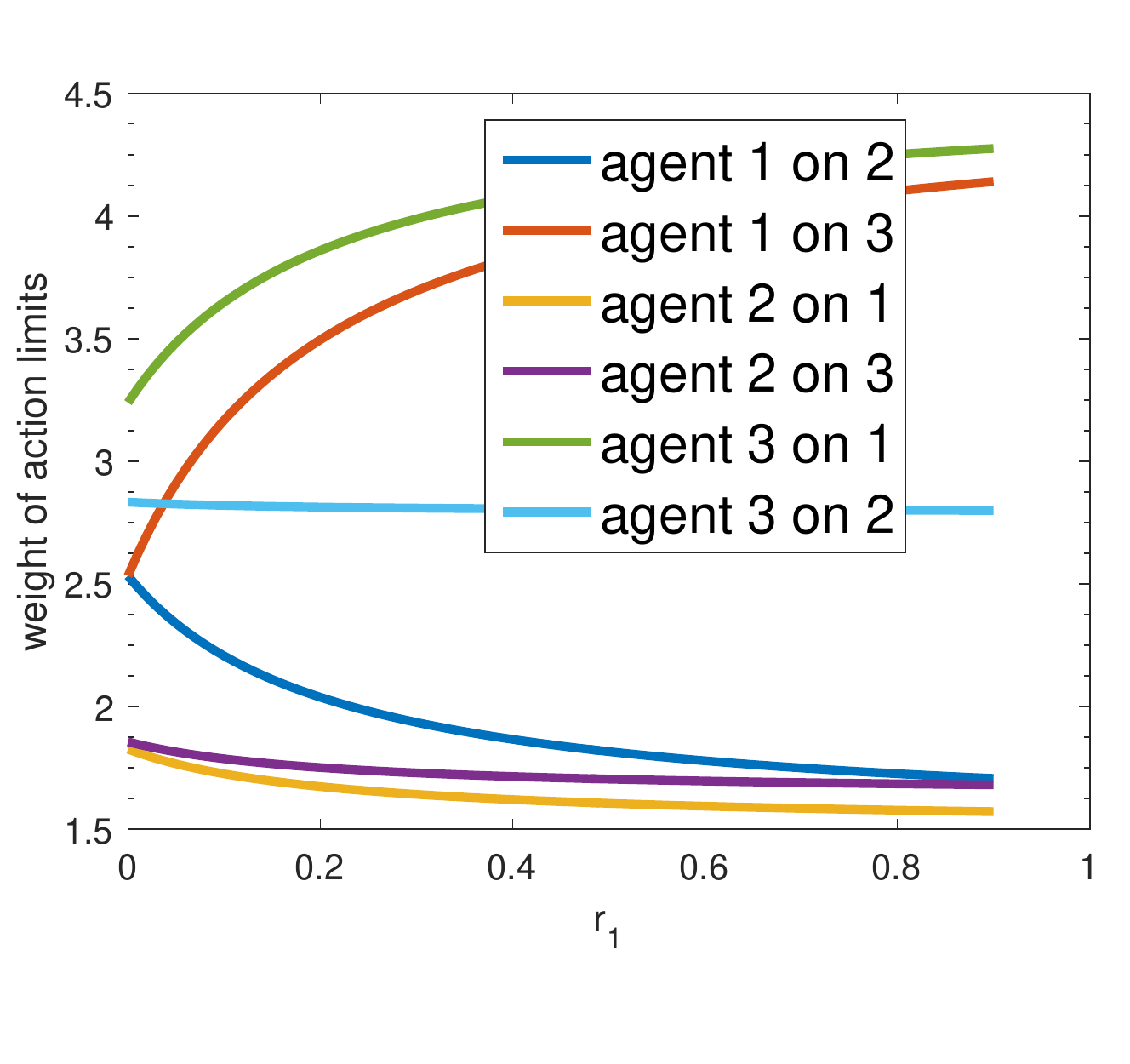}
}%
\protect\caption{Simulation results for the synchronous case, where the
limits of actions are plotted as functions of $r_1$, for
$r_2 = 0.1, r_3 = 0.6, r_1' = 0.1, r_2' = 0.4, r_3' = 0.1$,
$k_1 = 3, k_2 = 1, k_3 = 5$.
In the left graph, agent $1$ and $2$ are the only \name{fixed} agents,
while in the right one, $1$ is the only \name{floating} agent.
All the graphs exhibit monotonicity.
}
\label{fig:act_lim_r_1}
\end{figure}
%
%

%
\ifthenelse{\NOT \equal{\destin}{IJCAI16}}{
\begin{figure}[ht]
\centerline{%
\includegraphics[trim = 35mm 80mm 40mm 90mm, clip, width=0.24\textwidth, height=1.6in]{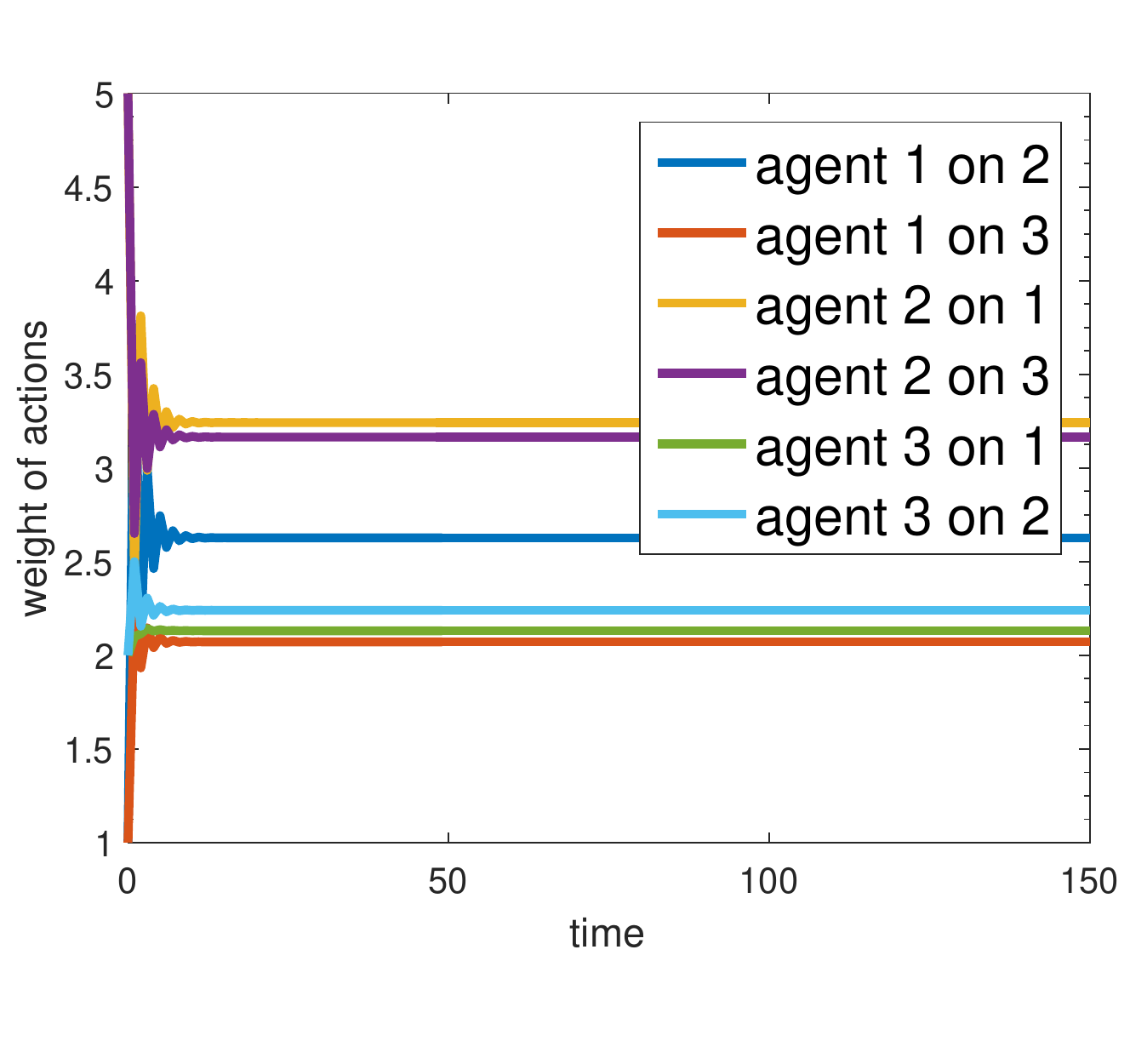}
}%
\protect\caption{Simulation results for the synchronous case, with three \name{fixed}
agents, for
$r_1 = 0.5, r_2 = 0.2, r_3 = 0.1, r_1' = 0.3, r_2' = 0.5, r_3' = 0.2$,
$k_1 = 1, k_2 = 5, k_3 = 2$.
We observe that $\lim_{t \to \infty}{x_{1,2}(t)} > \lim_{t \to \infty}{x_{3,2}(t)}$.
}
\label{fig:act_time_fi_fi_fi_A_1_5_2}
\end{figure}
}{
}%
\ifthenelse{\NOT \equal{\destin}{AAMAS16}}{
\section{Additional Notes}

When defining a reciprocating reaction, we used the last action of the
other agent to model the opinion about the other
agent.
We can explicitly define the \defined{opinion} of agent $i$ about another agent $j$
at time $t$, 
\ifthenelse{\equal{\content}{Two_agents}}{
$\opin_{i} \colon \reals^{t + 1} \to \reals$,
as
$\opin_{i}(t) \defas \imp_{j}(s_j(t))$,
}{
$\opin_{i, j} \colon \reals^{t + 1} \to \reals$,
as
$\opin_{i, j}(t) \defas \imp_{j, i}(s_j(t))$,
}
upon $i$.
Then, we obtain that
in the \name{fixed} reciprocation attitude%
\ifthenelse{\equal{\content}{Two_agents}}{
$\imp_i(t) \defas
\begin{cases}
(1 - r_i) \cdot k_i + r_i \cdot \opin_i(t - 1) & t > t_{i, 0}\\
k_i & t = t_{i, 0} = 0.
\end{cases}$
and in the \name{floating} reciprocation attitude
$\imp_i(t) =
\begin{cases}
(1 - r_i) \cdot \imp_i(s_i(t-1)) + r_i \cdot \opin_i(t - 1). 
 & t > t_{i, 0}\\
k_i & t = t_{i, 0} = 0.
\end{cases}$
}{
for $t > 0$, $\imp_{i,j}(t) \defas
(1 - r_i - r'_i) \cdot k_i + r_i \cdot \opin_{i, j}(t-1)
+ r'_i \cdot \frac{\got_i(t - 1)}{\abs{\Neighb(i)}}$.
and in the \name{floating} reciprocation attitude
for $t > 0$, $\imp_{i,j}(t) \defas
(1 - r_i - r'_i) \cdot \imp_{i, j}(s_i(t-1)) + r_i \cdot \opin_{i, j}(t-1) 
\\+ r'_i \cdot \frac{\got_i(t - 1)}{\abs{\Neighb(i)}}$.
}%

Naturally, a more general definition of opinion is possible.
To this end, we define the temporal distance in $T_i$, for an $i \in N$,
which designates how many times
agent $i$ acted between two given times in $T_i$. Formally, for
an $i \in N$ and two times
$t_{i, l}, t_{i, m} \in T_i$, we define $d_{T_i} \colon T_i^2 \to \realsP$ by
$d_{T_i}(t_{i, l}, t_{i, m}) \defas \abs{l - m}$. Now,
define the cumulative opinion of $i$ about $j$ at time $t$ to be%
\ifthenelse{\equal{\content}{Two_agents}}{
$\opin_{i}(t) \defas \sum_{t' \in T_j, t' \leq t}{\delta_i( d_{T_j}(t', s_j(t)) + 1) \cdot \imp_{j}(t')}$,
}{
$\opin_{i, j}(t) \defas \sum_{t' \in T_j, t' \leq t}{\delta_i( d_{T_j}(t', s_j(t)) + 1) \cdot \imp_{j, i}(t')}$,
}%
where $\delta_i(p) \colon \realsP \to \realsP$ is the discount function, expressing how much the
passed time influences the importance of an action.

Our definition of opinion as%
\ifthenelse{\equal{\content}{Two_agents}}{
$\opin_{i}({t}) =  \imp_{j}(s_j(t))$
}{
$\opin_{i, j}({t}) =  \imp_{j, i}(s_j(t))$
}%
is a particular case of this model, where the discount function
is $\delta_i(p) =
\begin{cases}
1 & p = 1,\\
0 & \text{otherwise}.
\end{cases}$

}{
}%
\section{Related Work}

In addition to the direct motivation for our model, presented in
Section~\ref{Sec:introd},
we were inspired by Trivers~\cite{Trivers1971} (a psychologist),
who describes a balance between an inner quality (immutable kindness)
and costs/benefits when determining an action. This idea of balancing
the inner and the outer appears also in our model.
\ifthenelse{\equal{\destin}{IJCAI16}}{
}{
Trivers also talks about a naturally selected complicated balance between altruistic
and cheating tendencies, which we model as kindness, which represents the inherent
inclination to contribute.
The balance between complying and not complying is mentioned
in the conclusion of~\cite{Ury2007}, motivating the convex combination between
own kindness or action and others' actions.
}%

The idea of humans behaving according to a convex combination resembles
another model, that of the altruistic extension, like%
~\cite{ChenElkindKoutsoupias2011,HoeferSkopalik2013,RahnSchafer2013},
and \chapt{iii.2} in~\cite{Ledyard1994}. 
In these papers, utility is often assumed being a convex combination,
while we consider a mechanism of an action being a convex combination.

\ifthenelse{\equal{\destin}{IJCAI16}}{
}{
Additional motivation stems from the bargaining and negotiation realm,
where
Pruitt~\cite{pruitt1981} mentions that in negotiation, cooperation often takes place
in the form of reciprocation and that personal traits influence the way of
cooperation, which corresponds in our model to the personal kindness and
reciprocation coefficients.
}%

\section{Conclusions and Future Work}


In order to facilitate behavioral decisions regarding reciprocation,
we need to predict what interaction a given setting will engender.
To this end, we model two reciprocation attitudes
where a reaction is a weighted combination of the action
of the other player, the total action of the neighborhood and either one's own kindness or one's own last action.%
\ifthenelse{\equal{\destin}{IJCAI16}}{
}{
This combination's weights are defined by the reciprocation coefficients.
}%
For a pairwise interaction,
we show that actions converge, find the exact limits,
%
and show that if you consider your
kindness while reciprocating (\name{fixed}), then, asymptotically, your
actions values get closer to your kindness than if you consider it only at the outset.
%
For a general network, we prove convergence and find the common limit
if all agents act synchronously and consider their last own action
(\name{floating}), besides at most one agent. Dealing with the case when multiple agents consider their
kindness (\name{fixed}) is mathematically hard, so we use simulations.%
\ifthenelse{\equal{\destin}{IJCAI16}}{
}{
We now substantiate these insights from
our results, beginning from the pairwise case.

For two agents with \name{fixed} reciprocation, (i.e.,
when a reaction partly depends on one's own kindness),
kinder agent's action are larger in the limit.
While interacting, 
each agent goes back and forth in her actions, monotonically narrowing to her limit.
Probably, this alternating may
make the process confusing for an outsider.

For two agents with \name{floating} reciprocation, (i.e.,
when a reaction partly depends on one's own last action),
both agents' actions converge to a common limit,
which vicinity to an agent's kindness is reversely proportional
to her reciprocation coefficient.
The commonality of the limit can intuitively result from an agent's
next action being a
combination of her last action with the other agent's last action,
which makes the new action closer to the other's action, this new action to be
taken into account in determining the next action.%
\ifthenelse{\NOT\equal{\destin}{IJCAI16}}{
This makes an agent aligned with the other agent.
}{
}%
%
\ifthenelse{\NOT\equal{\destin}{IJCAI16}}{
The behavior of the actions in the process depends on the sum of the
reciprocation coefficients. If the sum is at most $1$, then both actions converge
monotonically to the common limit, while otherwise, the action with
the larger weight becomes smaller at a time slot when both agents act
and the actions stay in the same relative order when only a single agent acts.
So, when agents are
not extremely cooperative, then the kinder agent acts stronger
all the time, but
if the agents strongly reciprocate to the other's behavior, then the agent
whose actions' weights are bigger switches
each time when both act. It is remarkable that when both agents act at all the
time slots, then for any given pair of
the parameters, the relative positions of the weights alternate at a given step if and only if the positions alternate
at all the steps.
}{
}%

For two agents,
when one agent acts according to the \name{fixed} reciprocation, and the other one
according to the \name{floating} one, both
actions converge to the kindness of the agent who employs \name{fixed} reciprocation.
This can be intuitively
explained as a result of one agent always considering her kindness in
determining the next action and thereby having a firm stance, while the other agent aligning himself.
%
\ifthenelse{\NOT\equal{\destin}{IJCAI16}}{
In the process, we know what happens only if the sum
of the reciprocation coefficients is at most $1$. In this case, they are both monotonic
from some time on,
and agent $2$'s action is always at least as large as $1$'s correspondent
action.
}{
}%
In~Example~\ref{ex:colleages} with two colleagues, the colleague who
ignores her inherent inclination and remembers only the last moves
will behave as the colleague who constantly considers her kindness.
Another conclusion is that if
the numerical parameters are set, then
the kinder agent employing \name{fixed} attitude
and the other one employing \name{floating} attitude is the best
for the total reciprocation.

When an agent may interact with any number of agents, we have proven convergence and shown that if
all agents employ \name{floating} reciprocation,
the limit is common. This limit is a weighted average of the kindness values,
the weight of an agent's kindness being her degree in the interaction graph divided
by the sum of her reciprocation coefficients. Intuitively, the agents
align to each other, and the more
connected and the less reciprocating an agent is, the more it influences
the common limit.
}%

In Example~\ref{ex:colleages} with the parameters from the end of
\sectnref{Sec:formal_model},
(all the agents
employ \name{floating} reciprocation),
\eqnsref{eq:all_float_lim} implies that all the actions approach~%
$25 / 52$ in the limit, meaning that all the colleagues support each other emotionally
a lot.

In addition to predicting the development of reciprocal interactions,
our results explain
why persistent agents have more influence on the interaction.
An expression of the converged behavior is that while growing up, people acquire their own style of reciprocating with
acquaintances~\cite{RobertsWaltonViechtbauer2006}.
In organizations, many styles are often very similar from
person to person, forming organizational cultures
~\cite{Hofstede1980}.

We saw in theory and we know from everyday life
that the reciprocation process may seem confusing, but the
exponential convergence promises the confusion to be short.
Actually, we can have a not so quick exponential convergence,
such as observed in the left graph in \figref{fig:act_time_fi_fl_fl_A_1_2_3_3_1_5},
but mostly, the process converges quickly.
Another important conclusion is that employing \name{floating} reciprocation
makes us achieve equality.
In the synchronous case, to achieve
a common limit
it is also enough for all the \name{fixed} agents to have the same kindness.
We also show that if all agents employ \name{floating} reciprocation
and act synchronously,
then the influence of an agent is proportional to her number of neighbors
and inversely proportional to her tendency to reciprocate, that is, the
stability.
We prove that in the synchronous case, the limit is either a linear
combination of the kindnesses of all the \name{fixed} agents or, if all
the agents are \name{floating}, a linear combination of the kindnesses of all the agents. Thus, an
agent's kindness influences nothing, or it is a linear factor, thereby
enabling a very eager agent to influence the limits arbitrarily,
by having the \name{fixed} attitude and the appropriate kindness.

As we see in examples, real situations may require more
complex modeling, motivating further research.%
For instance, modeling
interactions with a known finite time horizon would be
interesting.
Since people may change while reciprocating, modeling changes in the reciprocity coefficients
and/or reciprocation attitude is important.
In addition, groups of colleagues and nations get and lose people,
motivating modeling a dynamically changing set of reciprocating agents.
Even with the same set of agents, the interaction graph may change
as people move around.
\ifthenelse{\NOT \equal{\destin}{IJCAI16}}{
We have already presented some similar ideas from the negotiation
realm; therefore, we expect that bridging our work with negotiation can yield
many more insights.
}{
}%
We study interaction processes where agents
reciprocate with some given parameters, and show that maximizing $L$ would require extreme values of
reciprocation coefficients. To predict real situations better and to be able
to give constructive advice about what parameters and attitudes of the agents are
useful, we should define utility functions to the agents and consider
the game where agents choose their own parameters before the interaction
commences. This is hard, but people are able to change their behavior. %
\ifthenelse{\equal{\destin}{IJCAI16}}{
}{ The agents' strategic behavior may come at cost with respect to the
social welfare, so considering price of anarchy~\cite{KP99}
and stability~\cite{AnshelevichDasGuptaKleinbergTardosWexlerRoughgarden04}
of such a game is in order.
}%
Considering how to influence agents to change their behavior
is also relevant.
Though it seems extremely hard,
it would be nice to consider our model in the light of a
game theoretic model of an extensive form game,
such as~\cite{DufwenbergKirchsteiger2004}.
We used others' research, based on real data, as a basis for the model;
actually evaluating the model on relevant data, like the arms race
actions, may be enlightening.
An agent could have different kindness
values towards different agents, to represent her prejudgement. 
Another extension would be allowing the same action be perceived
differently by various agents. 
A system of agents who have both a \name{fixed} and a \name{floating}
component would be interesting to analyze.

\anote{Many more directions can be indicated, such as:
who is important to influence in the network,
evolutionary game theory with overtaking strategies, probabilistic reaction, etc.}

Analytical and simulations analysis of reciprocation process
allows estimating whether an interaction will be profitable to a given
agent and
lays the foundation for further modeling
and analysis of reciprocation, in order to anticipate and improve
the individual utilities and the social welfare.

\subsubsection*{Acknowledgments.}

This work has been supported by the project SHINE, the flagship project
of DIRECT (Delft Institute for Research on ICT at Delft University
of Technology).

\newpage

\bibliographystyle{abbrv}
\bibliography{library}




\end{document}